\documentclass[a4paper,UKenglish,cleveref, autoref, thm-restate]{lipics-v2021}

\title{Finding Induced Subgraphs from Graphs with Small Mim-Width}

\author{Yota Otachi}{Graduate School of Informatics, Nagoya University, Nagoya, Japan \and \url{https://www.math.mi.i.nagoya-u.ac.jp/~otachi/} }{otachi@nagoya-u.jp}{https://orcid.org/0000-0002-0087-853X}
{JSPS KAKENHI Grant Numbers
  JP18H04091, 
  JP20H05793, 
  JP21K11752, 
  JP22H00513. 
}

\author{Akira Suzuki}{Graduate School of Information Sciences, Tohoku University, Sendai, Japan \and \url{https://www.ecei.tohoku.ac.jp/alg/suzuki/}}{akira@tohoku.ac.jp}{https://orcid.org/0000-0002-5212-0202}{JSPS KAKENHI Grant Numbers JP20K11666, JP20H05794.}

\author{Yuma Tamura}{Graduate School of Information Sciences, Tohoku University, Sendai, Japan}{tamura@tohoku.ac.jp}{https://orcid.org/0009-0001-5479-7006}{JSPS KAKENHI Grant Number JP21K21278.}

\authorrunning{Y. Otachi, A. Suzuki, Y. Tamura} 

\Copyright{Yota Otachi, Akira Suzuki, Yuma Tamura} 

\ccsdesc[100]{Mathematics of computing~Graph algorithms} 

\keywords{mim-width, graph algorithm, NP-hardness, induced subgraph problem, cluster vertex deletion} 

\category{} 

\relatedversion{} 

\nolinenumbers 

\EventEditors{John Q. Open and Joan R. Access}
\EventNoEds{2}
\EventLongTitle{42nd Conference on Very Important Topics (CVIT 2016)}
\EventShortTitle{CVIT 2016}
\EventAcronym{CVIT}
\EventYear{2016}
\EventDate{December 24--27, 2016}
\EventLocation{Little Whinging, United Kingdom}
\EventLogo{}
\SeriesVolume{42}
\ArticleNo{23}

\usepackage{graphicx, xcolor}
\usepackage{mathrsfs}
\usepackage{tikz}
\usetikzlibrary{calc,positioning,shapes}
\usepackage{multicol}
\usepackage{enumerate}
\usepackage{multibib}
\usepackage{comment}
\newcites{app}{Additional references in the appendix}

\newcommand{\indG}[1]{G[#1]}
\newcommand{\Nei}[2]{N(#1; #2)}

\newcommand{\bij}{L}
\newcommand{\mimwe}[1]{\mathsf{mim}(#1)}
\newcommand{\mimwT}[2]{\mathsf{mimw}(#1, #2)}
\newcommand{\mimwG}[1]{\mathsf{mimw}(#1)}
\newcommand{\lmimwG}[1]{\mathsf{lmimw}(#1)}
\newcommand{\Prop}{\Pi}
\newcommand{\compProp}{\overline{\Prop}}
\newcommand{\seg}[2]{[#1, #2]}
\newcommand{\segsingle}[1]{[#1]}
\newcommand{\baseorder}{<}
\newcommand{\originorder}{\baseorder_{0}}
\newcommand{\chainorder}[1]{\baseorder_{#1}}
\newcommand{\nodecut}[1]{G_{#1, \overline{#1}}}
\newcommand{\head}[2]{\mathsf{head}_{#1}(#2)}
\newcommand{\tail}[2]{\mathsf{tail}_{#1}(#2)}
\newcommand{\heador}[2]{\mathsf{head}(#1, #2)}
\newcommand{\tailor}[2]{\mathsf{tail}(#1, #2)}

\newcommand{\eqrelation}[1]{\equiv_{#1}^{\mathsf{cl}}}
\newcommand{\eqrelationcc}[1]{\equiv_{#1}^{\mathsf{cc}}}
\newcommand{\eqrelationp}[1]{\equiv_{#1}^{\mathsf{po}}}
\newcommand{\eqrelationc}[1]{\equiv_{#1}^{\mathsf{con}}}

\newcommand{\eqrelationcp}[1]{\equiv_{#1}^{\mathsf{cp}}}

\newcommand{\rep}[1]{\mathsf{rep}_{#1}}

\newcommand{\ovl}[1]{\overline{#1}}
\newcommand{\clf}[3]{f_{#1}(#2,#3)}
\newcommand{\polf}[3]{f_{#1}^\mathsf{po}(#2,#3)}
\newcommand{\Repset}[1]{\mathscr{R}_{#1}}

\newcommand{\rev}[1]{#1}
\newcommand{\ONE}{\mbox{(I)}}
\newcommand{\TWO}{\mbox{(I\hspace{-.1em}I)}}
\newcommand{\THREE}{\mbox{{(I\hspace{-.1em}I\hspace{-.1em}I)}}}

\begin{document}

	\maketitle

	\begin{abstract}
        In the last decade, algorithmic frameworks based on a structural graph parameter called mim-width have been developed to solve generally NP-hard problems.
        However, it is known that the frameworks cannot be applied to the \textsc{Clique} problem, and the complexity status of many problems of finding dense induced subgraphs remains open when parameterized by mim-width.
        In this paper, we investigate the complexity of the problem of finding a maximum induced subgraph that satisfies prescribed properties from a given graph with small mim-width.
        We first give a meta-theorem implying that various induced subgraph problems are NP-hard for bounded mim-width graphs.
        Moreover, we show that some problems, including \textsc{Clique} and \textsc{Induced Cluster Subgraph}, remain NP-hard even for graphs with (linear) mim-width at most~$2$.
        In contrast to the intractability, we provide an algorithm that, given a graph and its branch decomposition with mim-width at most~$1$, solves \textsc{Induced Cluster Subgraph} in polynomial time.
        We emphasize that our algorithmic technique is applicable to other problems such as \textsc{Induced Polar Subgraph} and \textsc{Induced Split Subgraph}.
        Since a branch decomposition with mim-width at most~$1$ can be constructed in polynomial time for block graphs, interval graphs, permutation graphs, cographs, distance-hereditary graphs, convex graphs, and their complement graphs, our positive results reveal the polynomial-time solvability of various problems for these graph classes.
	\end{abstract}

 \section{Introduction}
	Efficiently solving intractable graph problems by using structural graph parameters has been extensively studied over the past few decades.
	Tree-width is arguably one of the most successful parameters in this research direction.
	Courcelle's celebrated result indicates that every problem expressible in $\textsf{MSO}_2$ logic is solvable in linear time for bounded tree-width graphs~\cite{Courcelle90}.
	Various graph problems, including \textsc{Independent Set}, \textsc{Clique}, \textsc{Dominating Set}, \textsc{Independent Dominating Set}, $k$-\textsc{Coloring} for a fixed $k$, \textsc{Feedback Vertex Set}, and \textsc{Hamiltonian Cycle}, can be written in $\textsf{MSO}_2$ logic, and hence Courcelle's theorem covers a wide range of problems.
    Later, Courcelle et al.~also gave an analogous result for a more general parameter than tree-width, namely, clique-width: every problem expressible in $\textsf{MSO}_1$ logic is solvable in linear time for bounded clique-width graphs (under the assumption that a $k$-expression for a fixed $k$ of an input graph is given)~\cite{CourcelleMR00}.
	However, these results are not applicable directly to problems on interval graphs and permutation graphs, because these graph classes have unbounded clique-width (and thus unbounded tree-width).

	In 2012, Vatshelle introduced mim-width~\cite{Vatshelle12}, and recently, algorithms based on mim-width have been widely developed~\cite{BelmonteV13,BergougnouxDJ23,BergougnouxK21,BergougnouxPT22,Bonomo-BrabermanBMP24,BrettellHMPP22,BrettellHMP22,BuixuanTV13,GalbyLR21,GalbyMR20,JaffkeKST19,JaffkeKT20a,JaffkeKT20b}.
	Roughly speaking, mim-width is an upper bound on the size of maximum induced matching along a branch decomposition of a graph.
	(In \Cref{sec:pre}, its formal definition will be given.)
	Mim-width is a more general structural parameter than clique-width in the sense that the class of bounded mim-width graphs properly contains the class of bounded clique-width graphs.
	Furthermore, many graph classes of unbounded clique-width have bounded mim-width: for example, interval graphs, permutation graphs, convex graphs, $k$-polygon graphs for a fixed $k$, circular $k$-trapezoid graphs for a fixed $k$, and $H$-graphs for a fixed graph $H$. (See~\cite{BelmonteV13,FominGR20} for more details.)
    Bergougnoux et al.~gave an algorithmic meta-theorem~\cite{BergougnouxDJ23}, which states that every problem expressible in \textsf{A}\&\textsf{C DN} logic is solvable in polynomial time for bounded mim-width graphs (under the assumption that a suitable branch decomposition of an input graph is given).
	\textsc{Independent Set}, \textsc{Dominating Set}, \textsc{Independent Dominating Set}, $k$-\textsc{Coloring} for a fixed $k$, \textsc{Feedback Vertex Set} etc.~can be expressed in \textsf{A}\&\textsf{C DN} logic.
    Thus, Bergougnoux et al.~showed that many problems are solvable in polynomial time for a much wider range of graph classes than the class of bounded clique-width graphs.

	Unfortunately, \textsf{A}\&\textsf{C DN} logic does not cover all problems expressible in $\textsf{MSO}_2$ logic.
	\textsc{Clique} and \textsc{Hamiltonian Cycle} cannot be written in \textsf{A}\&\textsf{C DN} logic, whereas they can be expressed in $\textsf{MSO}_1$ logic and $\textsf{MSO}_2$ logic, respectively.
    This means that the meta-theorem by Bergougnoux~et~al.~is not applicable to these problems.
	In fact, it is known that \textsc{Clique} is NP-hard for graphs with \rev{linear} mim-width\footnote{The linear mim-width of a graph $G$ is the mim-width when a branch decomposition of $G$ is restricted to a caterpillar. The formal definition will be given in \Cref{sec:pre}.} at most $6$~\cite{Vatshelle12} and \textsc{Hamiltonian Cycle} is NP-hard for graphs with \rev{linear} mim-width~$1$~\cite{JaffkeKT20a}.
    \rev{Note that by combining some known facts, we can show that \textsc{Clique} on graphs with mim-width at most 1 can be solved in polynomial time (see the discussion in the second paragraph of \Cref{sec:cluster}).}
	These results lead us to ask the following questions:
	\begin{itemize}
		\item What kind of problems expressible in $\textsf{MSO}_2$ logic are NP-hard for bounded mim-width graphs?
		\item Is \textsc{Clique} NP-hard for graphs with mim-width less than~$6$?
		\item Given a graph with mim-width at most~$1$, \rev{which $\textsf{MSO}_2$-expressible problems are polynomial-time solvable?}
	\end{itemize}

	\subsection{Our contributions}
	To answer the questions above, in this paper, we systematically study the complexity of the \textsc{Induced $\Prop$ Subgraph} problems and their \rev{complementary} problems, called the \textsc{$\Prop$ Vertex Deletion} problems, on bounded (linear) mim-width graphs.
	We first show that for any nontrivial hereditary graph property $\Prop$ that admits all cliques, \rev{there is a constant $w$ such that} \textsc{Induced $\Prop$ Subgraph} and \textsc{$\Prop$ Vertex Deletion} are NP-hard for \rev{graphs with (linear) mim-width at most $w$}.
	For example, \textsc{Clique}, \textsc{Induced Cluster Subgraph}, \textsc{Induced Polar Subgraph}, and \textsc{Induced Split Subgraph} satisfy the aforementioned conditions, and hence all of them are NP-hard for bounded (linear) mim-width graphs.
	As a byproduct, we also show that connected and dominating variants of them are NP-hard for bounded (linear) mim-width graphs.
	Moreover, we give sufficient conditions for \textsc{Induced $\Prop$ Subgraph} and \textsc{$\Prop$ Vertex Deletion} to be NP-hard for graphs with (linear) mim-width at most~$2$.
	\textsc{Clique}, \textsc{Induced Cluster Subgraph}, \textsc{Induced Polar Subgraph}, and \textsc{Induced Split Subgraph} are proven to be in fact NP-hard even for graphs with (linear) mim-width at most~$2$.
    We thus reveal that there are various NP-hard problems for bounded mim-width graphs, although they can be expressed in $\textsf{MSO}_2$ logic.
    Especially, our result for \textsc{Clique} strengthens the known result that \textsc{Clique} is NP-hard for graphs with mim-width at most $6$~\cite{Vatshelle12}.

	To complement the intractability, we next seek polynomial-time solvable cases for graphs with mim-width at most~$1$.
	Here we focus on \textsc{Induced Cluster Subgraph}, also known as \textsc{Cluster Vertex Deletion}.
	\textsc{Induced Cluster Subgraph} is known to be NP-hard for bipartite graphs~\cite{HsiehLLP24,Yannakakis81}, while it is solvable in polynomial time for split graphs, block graphs, interval graphs~\cite{Cao18}, cographs~\cite{LeL22}, bounded clique-width graphs~\cite{CourcelleMR00}, and convex graphs\footnote{If a given graph is convex (more generally $K_3$-free), \textsc{Induced Cluster Subgraph} is equivalent to \textsc{Induced $\Prop$ Subgraph} such that $\Prop$ is the class of graphs with maximum degree at most~$1$, which is solvable in polynomial time for convex graphs~\cite{BuixuanTV13}.}.
    Surprisingly, the complexity status of \textsc{Induced Cluster Subgraph} on chordal graphs is still open.
	We show that, given a graph $G$ with mim-width at most~$1$ accompanied by its branch decomposition with mim-width at most~$1$, \textsc{Induced Cluster Subgraph} is solvable in polynomial time.
    Although the complexity of computing a branch decomposition with mim-width at most~$1$ of a given graph is still open in general, our result yields a unified polynomial-time algorithm for \textsc{Induced Cluster Subgraph} that works on block graphs, interval graphs, permutation graphs, cographs, distance-hereditary graphs, convex graphs, and their complement graphs because all these graphs have mim-width at most~$1$ and their branch decompositions of mim-width at most~$1$ can be obtained in polynomial time~\cite{BelmonteV13,Oum05,Vatshelle12}.\footnote{As far as we know, it was not explicitly stated in any literature that block graphs and distance-hereditary graphs have mim-width at most~$1$. This follows from the facts that a graph is distance-hereditary if and only if its rank-width is at most~$1$~\cite{Oum05}, and block graphs are distance-hereditary graphs.}
	Consequently, we give \rev{independent} proofs for some of the results in~\cite{Cao18,LeL22} via mim-width.
	Moreover, to the best of our knowledge, this is the first polynomial-time algorithm for \textsc{Induced Cluster Subgraph} on permutation graphs.
	We also emphasize that our algorithmic technique can be applied to other problems such as \textsc{Induced Polar Subgraph}, \textsc{Induced Split Subgraph}, and so on.
    Combining our results, we give the complexity dichotomy of the above problems with respect to mim-width.

\subsection{Previous work on mim-width} \label{sec:previous_work}

    Mim-width is a relatively new graph structural parameter introduced by Vatshelle~\cite{Vatshelle12} and it has attracted much attention in recent years to design efficient algorithms of problems on graph classes that have unbounded tree-width and clique-width.
    Combined with the result of Belmonte and Vatshelle~\cite{BelmonteV13}, Bui-Xuan et al.\ provided XP algorithms of \textsc{Locally Checkable Vertex Subset and Vertex Partitioning} problems (\textsc{LC-VSVP} for short) parameterized by mim-width $w$, \rev{assuming that a branch decomposition with mim-width $w$ of a given graph can be computed in polynomial time}~\cite{BuixuanTV13}.
    Many problems, including \textsc{Independent Set}, \textsc{Dominating Set}, \textsc{Independent Dominating Set}, and \textsc{$k$-Coloring}, are expressible in the form of \textsc{LC-VSVP}.
    Jaffke et al.\ later generalized the result to the distance versions of \textsc{LC-VSVP}~\cite{JaffkeKST19}.
    As the name suggests, \textsc{LC-VSVP} can capture problems whose solutions are defined only by local constraints.
    \textsc{Longest Induced Path}~\cite{JaffkeKT20a} and \textsc{Feedback Vertex Set}~\cite{JaffkeKT20b} are the first problems with global constraints for which it was shown that there exist XP algorithms parameterized by mim-width.
    Bergougnoux and Kant\'{e} designed a framework to deal with problems with global constraints for bounded mim-width graphs~\cite{BergougnouxK21}.
    The remarkable meta-theorem given by Bergougnoux et al.\ is not only a generalization of all the above results in this section, but also a powerful tool for solving more complicated problems on bounded mim-width graphs~\cite{BergougnouxDJ23}.
    \textsc{Subset Feedback Vertex Set} is one of the few examples where there exists an XP algorithm parameterized by mim-width~\cite{BergougnouxPT22} although the meta-theorem does not work for it.

    Unfortunately, computing the mim-width of a given graph is W[1]-hard, and there is no polynomial-time approximation algorithm within constant factor unless $\text{NP} = \text{ZPP}$~\cite{SatherV16}.
    Even the complexity of determining whether a given graph has mim-width at most~$1$ is a long-standing open problem.
    Fortunately, it is known that various \rev{graph classes} have constant mim-width and their branch decompositions with constant mim-width are computable in polynomial time~\cite{BelmonteV13,BrettellHMPP22,BrettellHMP22,FominGR20,KangKST17,MunaroY23}.
    In particular, some famous graphs, such as block graphs, interval graphs, permutation graphs, cographs, distance-hereditary graphs, and convex graphs, have mim-width at most~$1$ and their branch decomposition with mim-width at most~$1$ can be obtained in polynomial time~\cite{BelmonteV13,Oum05}.
    The class of leaf power graphs, which is the more general class than interval graphs and block graphs, also have mim-width at most~$1$~\cite{JaffkeKT18}, although it is not known whether an optimal branch decomposition of a given leaf power graph can be obtained in polynomial time.
    On the other hand, the following graph classes have unbounded mim-width: strongly chordal split graphs~\cite{Mengel18}, co-comparability graphs~\cite{KangKST17,Mengel18}, circle graphs~\cite{Mengel18}, and chordal bipartite graphs~\cite{Brault-BaronCM15}.

    In contrast to a wealth of research on developing XP algorithms parameterized by mim-width and establishing lower and upper bounds on mim-width for specific graph classes, there has been limited research on the NP-hardness of problems for \rev{graph classes} with constant mim-width~\cite{JaffkeKT20a,JaffkeLS23,Vatshelle12}.

\section{Preliminaries} \label{sec:pre}

	Let $G = (V,E)$ be a graph.
    We assume that all the graphs in this paper are simple, undirected, and unweighted.
    We denote by $V(G)$ and $E(G)$ the vertex set and the edge set of $G$, respectively.
    We usually deal with undirected graphs.
    If a graph is directed, it will be explicitly stated.
    For two vertices $u,v$ of $G$, we denote by $uv$ an undirected edge joining $u$ and $v$, and denote by $(u,v)$ a directed edge from $u$ to $v$.
	For a vertex $v$ of $G$, we denote by $\Nei{G}{v}$ the \emph{(open) neighborhood} of $v$ in $G$, that is, $\Nei{G}{v}= \{ w \in V \mid vw \in E \}$.
	The \emph{degree} of a vertex $v$ of $G$ is the size of $\Nei{G}{v}$.
	For a vertex subset $V^\prime \subseteq V$, we denote by $\indG{V^\prime}$ the subgraph induced by $V^\prime$.
    We use the shorthand $G-V^\prime$ for $\indG{V\setminus V^\prime}$. 
    For positive integers $i$ and $j$ with $i \le j$, we write $\seg{i}{j}$ as the shorthand for the set $\{i, i+1, \ldots, j\}$ of integers.
	In particular, we write $\seg{1}{j} = \segsingle{j}$.

    For two graphs $G_1=(V_1,E_1)$ and $G_2=(V_2,E_2)$ with $V_1 \cap V_2 = \emptyset$, the \emph{disjoint union} of $G_1$ and $G_2$ is the graph whose vertex set is $V_1 \cup V_2$ and edge set is $E_1 \cup E_2$.
    For a graph $H$ and a positive integer $\ell$, $\ell H$ means the disjoint union of $\ell$ copies of $H$.
	The \emph{complement} of $G$, denoted by $\ovl{G}$, is the graph on the same vertex set $V(G)$ with the edge set $\{uv \mid u,v \in V(G), uv \notin E(G)\}$.
	An \emph{independent set} $I$ of $G$ is a vertex subset of $G$ such that any two vertices in $I$ are non-adjacent.
	A \emph{clique} $K$ of $G$ is a vertex subset of $G$ such that any two vertices in $K$ are adjacent.
	Obviously, an independent set of $G$ forms a clique of $\overline{G}$, and vice versa.
	A \emph{dominating set} $D$ of $G$ is a vertex subset of $G$ \rev{such that} $\Nei{G}{v} \cap D \neq \emptyset$ for every vertex $v \in V(G) \setminus D$.
	A graph $G$ is said to be \emph{connected} if there is a path between any two vertices of $G$.
	A maximal connected subgraph of $G$ is called a \emph{connected component} of $G$.
     A \emph{cut vertex} of $G$ is a vertex whose removal from $G$ increases the number of connected components.

	\subsection{Graph classes}
	A graph is \emph{bipartite} if its vertex set can be partitioned into two independent sets.
	For disjoint vertex sets $A$ and $B$ of a graph $G$, we denote by $\indG{A,B}$ the \emph{bipartite subgraph} with the vertex set $A\cup B$ and the edge set $\{ab \in E(G) \mid a \in A, b \in B\}$.
     A bipartite graph $G = (A\cup B, E)$ consisting of disjoint independent sets $A$ and $B$ is called a \emph{chain graph} if there is an ordering $ a_1, a_2, \dots, a_{|A|}$ of vertices in $A$ such that $\Nei{G}{a_1} \subseteq \Nei{G}{a_2} \subseteq \cdots \subseteq \Nei{G}{a_{|A|}}$.
     Note that, if $A$ has such an ordering, then $B$ also has an ordering $b_1, b_2, \dots, b_{|B|}$ of vertices in $B$ such that $\Nei{G}{b_1} \subseteq \Nei{G}{b_2} \subseteq \cdots \subseteq \Nei{G}{b_{|B|}}$.

	A \emph{tree} is a connected acyclic graph.
	A vertex of a tree is called a \emph{leaf} if it has degree~$1$; otherwise, it is an \emph{internal vertex}.
	A \emph{rooted tree} $T$ is a tree with a specific vertex $r$ called the \emph{root} of $T$.
    For a rooted tree $T$ and two adjacent vertices $x$ and $y$ of $T$, we say that $x$ is the \emph{parent} of $y$, and conversely, $y$ is a \emph{child} of $x$ if $x$ lies on a path from $y$ to $r$.
	A \emph{full binary tree} is a rooted tree such that each vertex has zero or exactly two children.
	A tree $T$ is a \emph{caterpillar} if it contains a path $P$ called a \emph{spine} such that every leaf of $T$ is adjacent to a vertex of $P$.
	In this paper, we assume that the spine $P$ is maximum, that is, there is no path longer than $P$.
	The vertices of degree at most~$1$ in $P$ are called the \emph{endpoints} of $P$.
	A tree $T$ is called \emph{subcubic} if every internal vertex of $T$ has degree exactly~$3$.

	We denote by $K_n$ and $P_n$ the complete graph and the path graph with $n$ vertices, respectively.
	We say that a graph $G$ is \emph{$H$-free} if $G$ does not contain a graph isomorphic to $H$ as an induced subgraph.
	More generally, $G$ is said to be \emph{$(H_1, H_2, \ldots)$-free} if $G$ contains none of $H_1, H_2, \ldots$ as induced subgraphs.
	A graph $G$ is a \emph{threshold} graph if and only if $G$ is $(2K_2, P_4, C_4)$-free~\cite{ChvatalH73}.

	\subsection{Mim-width}
	For an edge subset $E^\prime$ of a graph $G$, we denote $V(E^\prime) = \{v,w \in V(G) \mid vw \in E^\prime \}$.
	An edge subset $M\subseteq E(G)$ is an \emph{induced matching} of $G$ if every vertex of $\indG{V(M)}$ has degree exactly~$1$.
	For a vertex subset $A \subseteq V(G)$, let $\mimwe{A}$ be the maximum size of an induced matching in the bipartite subgraph $\indG{A, \ovl{A}}$, where $\ovl{A} = V(G) \setminus A$.

	A \emph{branch decomposition} of a graph $G$ is a pair $(T, \bij)$, where $T$ is a subcubic tree with $|V(G)|$ leaves and $\bij$ is a bijection from $V(G)$ to the leaves of $T$.
	In particular, a branch decomposition $(T, \bij)$ is called \emph{linear} if $T$ is a caterpillar.
	To distinguish vertices of $T$ from those of the original graph $G$, we call the vertices of $T$ \emph{nodes}.
	For each edge $e$ of $T$, as $T$ is acyclic, removing $e$ from $T$ results in two trees $T_1^e$ and $T_2^e$.
	Let $(A_1^e, A_2^e)$ be a vertex bipartition of $G$, where $A_i^e = \{\bij^{-1}(\ell) \mid  \text{$\ell$ is a leaf of $T_i^e$} \}$ for each $i \in \{1,2\}$.
	The \emph{mim-width} $\mimwT{T}{\bij}$ of a branch decomposition $(T, \bij)$ of $G$ is defined as $\max_{e\in E(T)} \mimwe{A_1^e}$.
	The \emph{mim-width} $\mimwG{G}$ of $G$ is the minimum mim-width over all branch decompositions of $G$.
	Similarly, the \emph{linear mim-width} $\lmimwG{G}$ of $G$ is the minimum mim-width over all linear branch decompositions of $G$.
	Note that $\mimwG{G} \le \lmimwG{G}$ holds for any graph $G$.

	In this paper, to make a branch decomposition easier to handle, we often consider its rooted variant.
	A \emph{rooted layout} of a graph $G$ is a pair $(T^\prime, \bij)$, where $T^\prime$ is a rooted full binary tree with $|V(G)|$ leaves and $\bij$ is a bijection from $V(G)$ to the leaves of $T^\prime$.
	The mim-width of a rooted layout $(T^\prime, \bij)$ is defined in the same way as a branch decomposition.
	A rooted layout of $G$ is obtained from a branch decomposition $(T, \bij)$ of $G$ with the same mim-width by inserting a root $r$ to an arbitrary edge of $T$.
	(If $|V(T)| = 1$, we regard the unique node of $T$ as the root $r$ of $T^\prime$.)

	Here we note propositions concerning mim-width.
	Vatshelle showed that for a graph $G$ and a vertex $v \in V(G)$, it holds that $\mimwG{G - {v}} \le \mimwG{G}$~\cite{Vatshelle12}.
    One can see that the proof given by Vatshelle suggests the next proposition.

	\begin{proposition} \label{prop:mimw_subgraph}
		For a graph $G$ and an induced subgraph $G^\prime$ of $G$, it holds that $\mimwG{G^\prime} \le \mimwG{G}$ and $\lmimwG{G^\prime} \le \lmimwG{G}$.
	\end{proposition}

	We here focus on graphs with mim-width at most~$1$.
	It is known that a graph $G$ is a chain graph if and only if $G$ is a bipartite graph with a maximum induced matching of size at most~$1$~\cite{HammerPS90}.
	Thus, we obtain the following proposition.
	\begin{proposition} \label{prop:mimone_chain}
		Let $(T,\bij)$ be a branch decomposition of a graph $G$.
		Then, $\mimwT{T}{\bij} \le 1$ if and only if for any edge $e$ of $T$, the bipartite subgraph $\indG{A_1^e, A_2^e}$ of $G$ is a chain graph.
	\end{proposition}

        Moreover, for a graph $G$ and a vertex subset $A \subset V(G)$, it is not hard to see that $\mimwe{A} \le 1$ on $G$ if and only if $\mimwe{A} \le 1$ on $\ovl{G}$ from the definition of a chain graph.
        This implies the following proposition.

	\begin{proposition}[\cite{Vatshelle12}] \label{prop:mimone_comp}
		Suppose that a graph $G$ has mim-width at most~$1$.
		Then, any branch decomposition $(T,\bij)$ of $G$ with $\mimwT{T}{\bij} \le 1$ is also the branch decomposition of $\ovl{G}$ with $\mimwT{T}{\bij} \le 1$.
		Consequently, $\mimwG{G} \le 1$ if and only if $\mimwG{\ovl{G}} \le 1$.
	\end{proposition}

	Combined with the observation that any induced cycle of length at least~$5$ has mim-width~$2$ and the strong perfect graph theorem~\cite{ChudnovskyRSR06}, \Cref{prop:mimone_comp} leads to the following proposition.

	\begin{proposition}[\cite{Vatshelle12}] \label{prop:mimone_perfect}
		All graphs with mim-width at most~$1$ are perfect graphs.
	\end{proposition}

	\subsection{Graph properties and problems}

	Let $\Prop$ be a fixed graph property.
	We often regard $\Prop$ as a collection of graphs satisfying the graph property.
	A graph property $\Prop$ is \emph{nontrivial} if there exist infinitely many graphs satisfying $\Prop$ and there exist infinitely many graphs
	that do not satisfy $\Prop$.
	A graph property $\Prop$ is said to be \emph{hereditary} if for any graph $G$ satisfying $\Prop$, every induced subgraph of $G$ also satisfies $\Prop$.
	We denote by $\compProp$ the \emph{complementary property} of $\Prop$, that is, $\compProp = \{\overline{G} : G \in \Prop \}$.

	For a graph $G$, a vertex subset $S \subseteq V(G)$ is called a \emph{$\Prop$-set} of $G$ if $\indG{S}$ satisfies $\Prop$.
	The \textsc{Induced $\Prop$ Subgraph} problem asks for \rev{a $\Prop$-set $S$ of maximum size for a given graph $G$}.
	If $\indG{S}$ is also required to be connected, then the problem is called the \textsc{Connected Induced $\Prop$ Subgraph} problem.
	For example, \textsc{Independent Set} is equivalent to \textsc{Induced $K_2$-free Subgraph}, and \textsc{Clique} is equivalent to \textsc{Induced $2K_1$-free Subgraph} and \textsc{Connected Induced $P_3$-free Subgraph}.
	Note that a vertex set $S$ of $G$ is a $\Prop$-set if and only if $S$ is a $\compProp$-set of $\overline{G}$.
	In \textsc{Induced $\Prop$ Subgraph}, if the $\Prop$-set $S$ is also required to be a dominating set of $G$, then the problem is called the \textsc{Dominating Induced $\Prop$ Subgraph} problem.

	Under the polynomial-time solvability, \textsc{Induced $\Prop$ Subgraph} is equivalent to the \textsc{$\Prop$ Vertex Deletion} problem, which asks for a minimum vertex subset $S^\prime$ of $G$ such that $G - S^\prime$ satisfies $\Prop$.
	The vertex subset $S^\prime$ is called a \emph{$\Prop$-deletion set} of $G$.
	The \textsc{Vertex Cover} problem is equivalent to \textsc{$K_2$-free Vertex Deletion}.
	If $\indG{S^\prime}$ is also required to be connected, then the problem is called the \textsc{Connected $\Prop$ Vertex Deletion} problem.
	In \textsc{$\Prop$ Vertex Deletion}, if the $\Prop$-deletion set $S^\prime$ is also required to be a dominating set of $G$, then the problem is called the \textsc{Dominating $\Prop$ Vertex Deletion} problem.

	\section{NP-hardness} \label{sec:hardness}

    \subsection{General case} \label{sec:hardness_general}

	In this section, we show the NP-hardness of \textsc{Induced $\Prop$ Subgraph} and \textsc{$\Prop$ Vertex Deletion} on \rev{graphs with linear mim-width at most $w$, where $w$ is some constant.}

	\begin{theorem}\label{the:NP-hard_main}
		Let $\Prop$ be a fixed nontrivial hereditary graph property \rev{that admits all cliques}.
        Then there is a constant $w$ such that \textsc{Induced $\Prop$ Subgraph} and \textsc{$\Prop$ Vertex Deletion}, as well as their connected variants and their dominating variants, are NP-hard for graphs with linear mim-width at most $w$, \rev{even if a branch decomposition with mim-width at most $w$ of an input graph is given.}
	\end{theorem}

    \rev{}
    Since \textsc{Induced $\Prop$ Subgraph} is the \rev{complementary} problem of \textsc{$\Prop$ Vertex Deletion}, we only prove the hardness of \textsc{$\Prop$ Vertex Deletion}.

    The \emph{girth} of a graph $G$ is the length of a shortest cycle in $G$.
	We reduce \textsc{Vertex Cover} on graphs with girth at least~$7$, which is known to be NP-complete~\cite{Poljak74}, to \textsc{$\Prop$ Vertex Deletion} by following the classical reduction technique of Lewis and Yannakakis~\cite{LewisY80}.

	First, we define a sequence on a graph.
	Consider a graph $H$ with $p$ connected components $H_1, H_2, \ldots, H_p$.
	Suppose that $H_i$ for $i \in \segsingle{p}$ has a cut vertex $c$ and the removal of $c$ from $H_i$ results in $q$ connected components $C_{i,1}, C_{i,2}, \ldots, C_{i,q}$ with $|V(C_{i,1})| \ge  |V(C_{i,2})| \ge \cdots \ge |V(C_{i,q})|$.
	For each $j \in \segsingle{q}$, we denote by $H_{i,j}$ the subgraph induced by $V(C_{i,j}) \cup \{c\}$ and $n_{i,j} = |V(H_{i,j})|$.
	The cut vertex $c$ gives a non-increasing sequence $\alpha_{c} = \langle n_{i,1}, n_{i,2}, \ldots, n_{i,q} \rangle$.
	For two sequences $\alpha_{c}$ and $\alpha_{c^\prime}$ according to cut vertices $c$ and $c^\prime$ of $H_i$, we write $\alpha_{c^\prime} <_L \alpha_{c}$ if $\alpha_{c^\prime}$ is smaller than $\alpha_{c}$ in the sense of lexicographic order.
	Let $\alpha_i$ be the lexicographically smallest sequence among all sequences according to the cut vertices of $H_i$.
	If $H_i$ has no cut vertex, we let $\alpha_i =  \langle |V(H_i)| \rangle$.
	Define $\beta_{H} =  \langle \alpha_1, \alpha_2, \ldots, \alpha_p \rangle$, where we assume that $\alpha_1 \ge_L \alpha_2 \ge_L \cdots \ge_L \alpha_p$.
	For example, for the graph $H$ depicted in \figurename~\ref{fig:forbidden_subgraph}(a), we have $\beta_{H} =  \langle  \langle 4,2 \rangle, \langle 2,2 \rangle \rangle$.
	For two graphs $H$ with $p$ connected components and $H^\prime$ with $q$ connected components, we write $\beta_{H^\prime} <_R \beta_{H}$ if $\beta_{H^\prime}$ is smaller than $\beta_{H}$ in the sense of lexicographic order: more precisely, assuming that  $\beta_{H} =\langle \alpha_1, \alpha_2, \ldots, \alpha_p \rangle$ and $\rev{\beta_{H^\prime}} =\langle \alpha_1^\prime, \alpha_2^\prime, \ldots, \alpha_q^\prime \rangle$, there exists an integer $i \in \segsingle{\min \{p,q\}}$ such that $\alpha_j^\prime = \alpha_j$ for every $j \in \segsingle{i-1}$ and $\alpha_i^\prime <_L \alpha_i$; or $q < p$ and $\alpha_i^\prime = \alpha_i$ for every $i \in \segsingle{q}$.

	Consider the complementary property $\compProp$ of $\Prop$.
	Note that, since all cliques satisfy $\Prop$, all independent sets satisfy $\compProp$.
	Let $F$ be a graph satisfying the following two conditions:
	\begin{enumerate}
		\item there is an integer $\ell \ge 1$ such that $\ell F$ violates $\compProp$, whereas $(\ell-1) F$ satisfies $\compProp$; and
		\item for any integer $\ell^\prime \ge 1$ and any graph $F^\prime$ with $\beta_{F^\prime} <_R \beta_{F}$, $\ell^\prime F^\prime$ satisfies $\compProp$.
	\end{enumerate}
	We call $F$ the \emph{base of $\compProp$-forbidden subgraphs}.
	Notice that the existence of $F$ is guaranteed because $\compProp$ is nontrivial.
	Moreover, $F$ and $\ell$ depend on $\Prop$ solely and are independent of an instance of \textsc{Vertex Cover}, that is, $F$ and $\ell$ are fixed.

	Let $F_1, F_2, \ldots, F_p$ be $p$ connected components of $F$, where $\alpha_1 \ge_L \alpha_2 \ge_L \cdots \ge_L \alpha_p$.
	We denote by $c_1$ the cut vertex of $F_1$ that realizes $\alpha_1$ (see \figurename~\ref{fig:forbidden_subgraph}(a)) and by $F_{1,1}$ the induced subgraph of $F_1$ corresponding to $n_{1,1}$ (see \figurename~\ref{fig:forbidden_subgraph}(b)).
    If $F_1$ has no cut vertex, then $c_1$ is any vertex of $F_1$.
	We then arbitrarily choose a vertex from $\Nei{F_{1,1}}{c_1}$ and label it as $d$.
	Notice that $\Nei{F_{1,1}}{c_1} \neq \emptyset$; otherwise, since $\alpha_1, \alpha_2, \ldots, \alpha_p$ are lexicographically sorted, $\ell F$ \rev{is} an independent set and violates $\compProp$, which contradicts that all independent sets satisfy $\compProp$.
	Let $F^\prime$ be the graph obtained by removing $V(F_{1,1})\setminus \{c_1\}$ from $F$ (see \figurename~\ref{fig:forbidden_subgraph}(c)).

	\begin{figure}[t]
		\centering
		\begin{tabular}{ccc}

			\begin{minipage}[t]{0.25\linewidth}
				\centering
				\begin{tikzpicture}
					\def\edgew{0.5}
					\def\gspace{1.5}
                    \def\nodew{2.5}

					\node[draw, circle, very thick, fill = black, inner sep=\nodew] at (0:\edgew) (v1){};
					\node[draw, circle, very thick, label=180:{$c_1$}, fill = black, inner sep=\nodew] at (90:\edgew) (v2){};
					\node[draw, circle, very thick, fill = black, inner sep=\nodew] at (180:\edgew) (v3){};
					\node[draw, circle, very thick, fill = black, inner sep=\nodew] at (270:\edgew) (v4){};
					\node[draw, circle, very thick, fill = black, inner sep=\nodew] at ($(v2) + (0,1.5*\edgew)$) (v5){};

					\draw [very thick] (v1) -- (v2) -- (v3) -- (v4) -- (v1);
					\draw [very thick] (v2) -- (v5);

					\node[draw, circle, very thick, fill = black, inner sep=\nodew] at ($(v4) + (\gspace,0)$) (w1){};
					\node[draw, circle, very thick, fill = black, inner sep=\nodew] at ($(v5) + (\gspace,0)$) (w3){};
					\node[draw, circle, very thick, label=0:{$c_2$}, fill = black, inner sep=\nodew] at ($(w1)!0.5!(w3)$) (w2){};

					\draw [very thick] (w1) -- (w2) -- (w3);
				\end{tikzpicture}

			\end{minipage} &




			\begin{minipage}[t]{0.25\linewidth}
				\centering
				\begin{tikzpicture}
					\def\edgew{0.5}
					\def\gspace{2}
                    \def\nodew{2.5}

					\node[draw, circle, very thick, fill = black, inner sep=\nodew] at (0:\edgew) (v1){};
					\node[draw, circle, very thick, label=180:{$c_1$}, fill = black, inner sep=\nodew] at (90:\edgew) (v2){};
					\node[draw, circle, very thick, label=180:{$d$}, fill = black, inner sep=\nodew] at (180:\edgew) (v3){};
					\node[draw, circle, very thick, fill = black, inner sep=\nodew] at (270:\edgew) (v4){};

					\draw [very thick] (v1) -- (v2) -- (v3) -- (v4) -- (v1);
				\end{tikzpicture}
			\end{minipage} &




			\begin{minipage}[t]{0.25\linewidth}
				\centering
				\begin{tikzpicture}
					\def\edgew{0.45}
					\def\gspace{1.25}
                    \def\nodew{2.5}

					\node[circle, very thick] at (0:\edgew) (v1){};
					\node[draw, circle, very thick, label=180:{$c_1$}, fill = black, inner sep=\nodew] at (90:\edgew) (v2){};
					\node[circle, very thick] at (180:\edgew) (v3){};
					\node[circle, very thick] at (270:\edgew) (v4){};
					\node[draw, circle, very thick, fill = black, inner sep=\nodew] at ($(v2) + (0,1.5*\edgew)$) (v5){};

					\draw [very thick] (v2) -- (v5);

					\node[draw, circle, very thick, fill = black, inner sep=\nodew] at ($(v4) + (\gspace,0)$) (w1){};
					\node[draw, circle, very thick, fill = black, inner sep=\nodew] at ($(v5) + (\gspace,0)$) (w3){};
					\node[draw, circle, very thick, fill = black, inner sep=\nodew] at ($(w1)!0.5!(w3)$) (w2){};

					\draw [very thick] (w1) -- (w2) -- (w3);
				\end{tikzpicture}
			\end{minipage}
			\\
			(a) & ~~~(b) & (c)
		\end{tabular}
		\caption{Let $F$ be the graph depicted in~(a). The cut vertex $c_1$ of the left connected component $F_1$ of $F$ gives $\alpha_1 = \langle 4, 2 \rangle$  and the cut vertex $c_2$ of the right connected component $F_2$ of $F$ gives $\alpha_2 = \langle 2, 2 \rangle$, where $\alpha_1 >_L \alpha_2$. Thus, if $F$ is selected as the base of $\compProp$-forbidden subgraphs, $F_{1,1}$ and $F^\prime$ are defined as the graphs depicted in~(b) and~(c), respectively.}
		\label{fig:forbidden_subgraph}
	\end{figure}
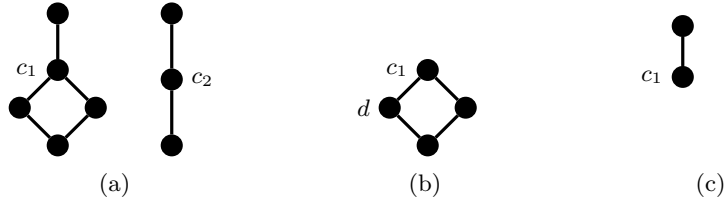

	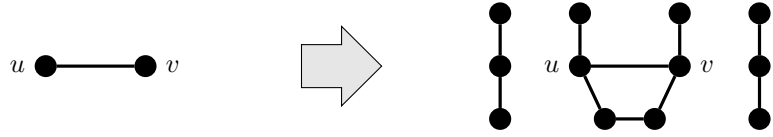
\begin{figure}[t]
		\centering
		\begin{tikzpicture}
			\def\edgew{1.75}
			\def\vheight{0.7}
           \def\nodew{2.5}

			\node[draw, circle, very thick, label=180:{$u$}, fill = black, inner sep=\nodew] at (0,0) (v1){};
			\node[draw, circle, very thick, label=0:{$v$}, fill = black, inner sep=\nodew] at (0.75*\edgew,0) (v2){};

			\draw [very thick] (v1) -- (v2);

			\node at (3.75, 0)[single arrow,draw=black,fill=black!10,minimum height=3em, minimum width=3em, single arrow head extend=0.3em] (arrow){};

			\begin{scope}[xshift=20em]
				\node[draw, circle, very thick, label=180:{$u$}, fill = black, inner sep=\nodew] at (0,0) (w1){};
				\node[draw, circle, very thick, label=0:{$v$}, fill = black, inner sep=\nodew] at (0.75*\edgew,0) (w2){};
				\node[circle, very thick] at ($(w1)!0.25!(w2)$) (x1){};
				\node[circle, very thick] at ($(w1)!0.75!(w2)$) (x2){};
				\node[draw, circle, very thick, fill = black, inner sep=\nodew] at ($(x1) + (0,-\vheight)$) (y1){};
				\node[draw, circle, very thick, fill = black, inner sep=\nodew] at ($(x2) + (0,-\vheight)$) (y2){};
				\draw [very thick, fill = black, inner sep=\nodew] (w1) -- (w2) -- (y2) -- (y1) -- (w1);

				\node[draw, circle, very thick, fill = black, inner sep=\nodew] at ($(w1) + (0,\vheight)$) (z1){};
				\node[draw, circle, very thick, fill = black, inner sep=\nodew] at ($(w2) + (0,\vheight)$) (z2){};
				\draw [very thick, fill = black] (w1) -- (z1);
				\draw [very thick, fill = black] (w2) -- (z2);
				\node[draw, circle, very thick, fill = black, inner sep=\nodew] at ($(z1) + (-0.6*\edgew,0)$) (u11){};
				\node[draw, circle, very thick, fill = black, inner sep=\nodew] at ($(w1) + (-0.6*\edgew,0)$) (u12){};
				\node[draw, circle, very thick, fill = black, inner sep=\nodew] at ($(w1) + (-0.6*\edgew,-\vheight)$) (u13){};
				\draw [very thick, fill = black] (u11) -- (u12) -- (u13);
				\node[draw, circle, very thick, fill = black, inner sep=\nodew] at ($(z2) + (0.6*\edgew,0)$) (u21){};
				\node[draw, circle, very thick, fill = black, inner sep=\nodew] at ($(w2) + (0.6*\edgew,0)$) (u22){};
				\node[draw, circle, very thick, fill = black, inner sep=\nodew] at ($(w2) + (0.6*\edgew,-\vheight)$) (u23){};
				\draw [very thick, fill = black] (u21) -- (u22) -- (u23);
			\end{scope}

		\end{tikzpicture}
		\caption{A transformation of an edge $uv$ with $F_{1,1}$ and $F^\prime$, which are the graphs depicted in \figurename~\ref{fig:forbidden_subgraph}(b) and~(c), respectively.}
		\label{fig:instance_trans}
	\end{figure}

	We now construct an input graph $G$ for \textsc{$\Prop$ Vertex Deletion} from an input graph $H$ with girth at least~$7$ for \textsc{Vertex Cover}.
	Let $n = |V(H)|$ and $H^\ast$ be the disjoint union of $\ell n$ copies of $H$.
	\rev{We assume that $n\ge 2$, $k < n-1$, and $H$ has at least one edge; otherwise, \textsc{Vertex Cover} is trivially solvable.}
	For each vertex $u$ of $H^\ast$, make a copy of $F^\prime$ and identify $c_1$ with $u$.
	For each edge $uv$ of $H^\ast$, make a copy of $F_{1,1}$ and identify $c_1$ and $d$ with $u$ and $v$, respectively.
	(See \figurename~\ref{fig:instance_trans}.)
    \rev{Let $H^\prime$ be the graph resulting from} the above transformation.
	Finally, we let $G = \overline{H^\prime}$.
	Since $\ell$, $F^\prime$, and $F_{1,1}$ are fixed, $G$ can be constructed in polynomial time in the size of $H$.

	In~\cite{LewisY80}, it is shown that $H$ has a vertex cover of size at most $k$ if and only if $H^\prime$ has a $\compProp$-deletion set $S$ of size at most $k\ell n$.
    \rev{Notice that $H^\prime$ has $\ell n \ge 2$ connected components because $H^\prime$ is obtained from $H^\ast$, which is the disjoint union of $\ell n$ copies of $H$.
    Moreover, we have the following lemma.}

    \begin{lemma} \label{lem:deletion_set_components}
        Suppose that $H^\prime$ has a $\compProp$-deletion set $S$ of size at most $k\ell n$.
        Then the following two claims (a) and (b) are true:
        \begin{enumerate}[(a)]
            \item there are two connected components $C_1$ and $C_2$ of $H^\prime$ such that $V(C_1) \setminus S \neq \emptyset$ and $V(C_2) \setminus S \neq \emptyset$; and
            \item there are two connected components $C_1^\prime$ and $C_2^\prime$ of $H^\prime$ such that $ V(C_1^\prime) \cap S  \neq \emptyset$ and $V(C_2^\prime) \cap S \neq \emptyset$.
        \end{enumerate}
    \end{lemma}

    \begin{proof}
        In the claim (a), assume for a contradiction that there is at most one connected component $C$ of $H^\prime$ such that $V(C) \setminus S \neq \emptyset$.
        In other words, $V(H^\prime-C) \subseteq S$ holds.
        Recall that $k < n-1$ and $H^\prime$ has $\ell n \ge 2$ connected components.
        Moreover, each connected component of $H^\prime - C$ has at least $n$ vertices from the construction of $H^\prime$.
        Thus, we have
        \[
        |S| \ge  |V(H^\prime -C)| \ge  n ( \ell n - 1) > (k+1) (\ell n -1)  = k\ell n + \ell n -k -1  > k\ell n,
        \]
        a contradiction.

        To prove the claim (b), assume for a contradiction that there is at most one connected component $C^\prime$ of $H^\prime$ such that $V(C^\prime) \cap S \neq \emptyset$.
        In other words, there are at least $\ell n -1$ ($ \ge \ell$ because $n\ge 2$) connected components of $H^\prime - C^\prime$ that contain no vertex in $S$.
        Consider $\ell$ connected components of $H^\prime - C^\prime$.
        Since each of them contains $F$ as an induced subgraph, $H^\prime - C^\prime$ contains $\ell F$ as an induced subgraph.
        However, $\ell F$ violates $\compProp$ because $F$ is the base of $\compProp$-forbidden subgraphs.
        This contradicts that $S$ is a $\compProp$-deletion set of $H^\prime$.
    \end{proof}

	\rev{Observe that $S$ is a $\compProp$-deletion set of $H^\prime$ of size at most $k\ell n$ if and only if $S$ is a $\Prop$-deletion set of $G = \overline{H^\prime}$ of size at most $k\ell n$.
    Combined with \Cref{lem:deletion_set_components}, this implies that $H$ has a vertex cover of size at most $k$ if and only if $G$ has a $\Prop$-deletion set $S$ of size at most $k\ell n$ such that the induced subgraphs $\indG{S}$ and $G-S$ are both connected, and $S$ and $V(G)\setminus S$ are dominating sets of $G$.}

	Our remaining task is to show that $G$ has \rev{linear mim-width at most $w$ for some constant $w$}.
	To this end, we consider a sequence of subgraphs of $H^\prime$.
	Let $V(H^\ast) = \{v_1, v_2, \ldots, v_n\}$ and $E(H^\ast) = \{e_1, e_2, \ldots, e_m\}$, where $n$ and $m$ are the numbers of vertices and edges of $H^\ast$, respectively.
	We make a sequence $\mathcal{H} = \langle H^\ast = H_0, H_1, \ldots, H_{n+m}=H^\prime \rangle$ such that $H_{i}$ for $i \in \segsingle{n}$ is obtained from $H_{i-1}$ by attaching a copy of $F^\prime$ to $v_i$, and $H_i$ for $i \in \seg{n+1}{n+m}$ is obtained from $H_{i-1}$ by attaching a copy of $F_{1,1}$ to $e_{i-n}$.
	Since $G = \overline{H^\prime}$, the following lemma completes the proof of \Cref{the:NP-hard_main}.
	(Recall that $F$ is fixed and hence $\lmimwG{\overline{F}}$ is a constant.)

	\begin{lemma}\label{lem:girth_mimw}
		For any graph $H_i$ in the sequence $\mathcal{H} = \langle H^\ast = H_0, H_1, \ldots, H_{n+m}=H^\prime \rangle$, \rev{a linear branch decomposition of $H_i$ with mim-width at most $\lmimwG{\overline{F}}+2$ can be obtained in polynomial time in the size of $H^\ast$.}
	\end{lemma}

	\begin{proof}
		We prove the lemma by induction, where the base case is $H_0 = H^\ast$.
		Note that $H^\ast$ has girth at least~$7$ because $H^\ast$ consists of copies of $H$ whose girth is at least~$7$.
		Consider a linear branch decomposition $(T_0,\bij_0)$ of $\overline{H_0}$, where $\bij_0$ is an arbitrary bijection from $V(H_0)$ to the leaves of $T_0$.
		To show that $\mimwT{T_0}{\bij_0} \le 2 \le \lmimwG{\overline{F}}+2$, assume for a contradiction that there is an edge $e$ of $T_0$ such that $\mimwe{A_1^e} \ge 3$ for the bipartition $(A_1^e, A_2^e)$.
		Let $x_1x_2, y_1y_2, z_1z_2$ be edges that form an induced matching in $\indG{A_1^e, A_2^e}$, where $x_1, y_1, z_1 \in A_1^e$ and $x_2, y_2, z_2 \in A_2^e$.
		Then, $x_1y_2, y_2z_1, z_1x_2, x_2y_1, y_1z_2, z_2x_1 \notin E(\overline{H})$ and hence they form a cycle of length~$6$ in $H_0$.
		This contradicts that the girth of $H_0$ is at least~$7$.

		Consider the case of $i > 0$.
		We here define a \emph{concatenation} of two linear branch decompositions.
		Let $G_1$ and $G_2$ be vertex-disjoint induced subgraphs of a graph $G$ such that $V(G_1) \cup V(G_2) = V(G)$, and let $(T_1,\bij_1)$ and $(T_2,\bij_2)$ be linear branch decompositions of $G_1$ and $G_2$, respectively.
		A concatenation of $(T_1,\bij_1)$ and $(T_2,\bij_2)$ is to construct a new linear branch decomposition $(T,\bij)$ of $G$ as follows.
        For each $i \in \{1,2\}$, let $e_i$ be an edge incident to an endpoint of the spine of $T_i$.
		Insert nodes $t_1$ and $t_2$ into $e_1$ and $e_2$, respectively, and then connect $t_1$ and $t_2$ by an edge.
		(If $|V(T_i)| = 1$ for $i \in \{1,2\}$, we define $t_i$ as the unique node of $T_i$.)
		Observe that $T$ is a subcubic caterpillar.
		Finally, set a bijection $\bij$ from $V(G)$ to the leaves of $T$ such that $\bij(v) = \bij_1(v)$ if $v \in V(G_1)$ and $\bij(v) = \bij_2(v)$ if $v \in V(G_2)$.

		By the induction hypothesis, there exists a linear branch decomposition $(T_{i-1},\bij_{i-1})$ of $\overline{H_{i-1}}$ such that $\mimwT{T_{i-1}}{\bij_{i-1}} \le \lmimwG{\overline{F}}+2$.
		Recall that $H_i$ is constructed by attaching a copy of $F^\prime$ to $v_i$ or a copy of $F_{1,1}$ to $e_{i-n}$.
		We denote by $F_i$ the subgraph of $H_i$ obtained by removing all vertices in $V(H_{i-1})$.
		We may assume that $|V(F_i)| \ge 1$; otherwise, $H_i = H_{i-1}$ and thus we immediately conclude that $\lmimwG{\overline{H_i}} \le \lmimwG{\overline{F}}+2$.
		Let $(T^\prime_i,\bij^\prime_i)$ be a linear branch decomposition of $\overline{F_i}$ such that $\mimwT{T^\prime_i}{\bij^\prime_i} \le \lmimwG{\overline{F}}$.
		Notice that, since $\overline{F_i}$ is an induced subgraph of $\overline{F}$, such a linear branch decomposition exists by \Cref{prop:mimw_subgraph}.
        \rev{Moreover, it can be constructed in constant time because $\ovl{F}$ is fixed.}
		We define $(T_i, \bij_i)$ as a linear branch decomposition obtained by a concatenation of $(T_{i-1},\bij_{i-1})$ and $(T^\prime_i,\bij^\prime_i)$.
       \rev{Clearly, the construction of $(T_i, \bij_i)$ can be done in polynomial time in the size of $H^\ast$.}

		To show that $\mimwT{T_{i}}{\bij_{i}} \le \lmimwG{\overline{F}}+2$, assume for a contradiction that there is an edge $e$ of $T_i$ such that the bipartite subgraph $\indG{A_1^e, A_2^e}$ of $\ovl{H_i}$ has an induced matching $M$ of size $\lmimwG{\overline{F}}+3$, where $(A_1^e, A_2^e)$ is the bipartition of $V(H_i)$ given by $e$.
		From the construction of $(T_i, \bij_i)$, the following two cases are considered: \ONE\ $A_1^e \subseteq V(H_{i-1})$ and $V(F_i) \subseteq A_2^e$; and \TWO\ $A_1^e \subseteq V(F_i)$ and $V(H_{i-1}) \subseteq A_2^e$.

		\smallskip
		\noindent \textbf{Case~\ONE:}
		Let $e^\prime$ be an edge of $T_{i-1}$ such that $A_1^{e^\prime} =  A_1^e$ and $A_2^{e^\prime} = A_2^e \setminus V(F_i)$.
		If $V(M) \subseteq V(H_{i-1})$, then $M$ is also an induced matching of the bipartite subgraph $\indG{A_1^{e^\prime}, A_2^{e^\prime}}$ defined by the linear branch decomposition $(T_{i-1}, \bij_{i-1})$.
		This implies that $\mimwT{T_{i-1}}{\bij_{i-1}} \ge |M| = \lmimwG{\overline{F}}+3$, which contradicts that $\mimwT{T_{i-1}}{\bij_{i-1}} \le \lmimwG{\overline{F}}+2$.

		Without loss of generality, we assume that $M$ has three distinct edges $x_1x_2, y_1y_2, z_1z_2$ such that $x_1, y_1, z_1 \in A_1^e \subseteq V(H_{i-1})$, $x_2 \in A_2^e \cap V(F_i)$, and $y_2, z_2 \in A_2^e$.
		Then, the sequence $\langle x_2, y_1, z_2, x_1, y_2, z_1, x_2 \rangle$ of vertices forms a cycle $C$ of length~$6$ of $H_i$.
		If $i \in \segsingle{n}$, as $x_2 \in V(F_i)$ is adjacent to at most one vertex in $V(H_{i-1})$ from the construction of $H_i$, then we have $y_1 = z_1$, a contradiction.
        Suppose that  $i \in \seg{n+1}{n+m}$.
        Recall that, from the construction of $H_i$, each vertex of $F_i$ is not adjacent to vertices in $V(H_{i-1})$ except for the endpoints of $e_{i-n}$.
		Since $x_2 \in V(F_i)$ is adjacent to the distinct vertices $y_1, z_1 \in V(H_{i-1})$, we have $e_{i-n} = y_1z_1$.
		Furthermore, $x_1 \in V(H_{i-1})$ is not adjacent to any vertex in $V(F_i)$ and hence we have $y_2, z_2 \in V(H_{i-1})$.
		Therefore, we obtain the cycle $C_1 = \langle y_1, z_2, x_1, y_2, z_1, y_1 \rangle$ with smaller length than that of $C$, where the vertices of $C_1$ are in $V(H_{i-1})$.
		Similarly, if $C_1$ contains vertices of $F_j$ for $j \in \seg{n+1}{i}$, there exists a smaller cycle of $H_{i-1}$ that contains no vertices of $F_j$.
		We eventually obtain a cycle $C^\prime$ of $H$ of length less than~$6$, which contradicts that $H$ has girth at least~$7$.

		\smallskip
		\noindent \textbf{Case~\TWO:}
        Recall that at most two vertices in $V(H_{i-1})$, say $u$ and $w$, are adjacent to some vertex in $V(F_i)$ on $H_i$ and thus no vertex in $V(H_{i-1}) \setminus \{u,w\}$ is adjacent to any vertex in $V(F_i)$ on $H_i$.
        If some vertex in $V(M)$ is in $V(H_{i-1}) \setminus \{u,w\}$, then we can take three distinct edges $x_1x_2, y_1y_2, z_1z_2 \in M$ such that $x_1, y_1, z_1 \in A_1^e \subseteq V(F_i)$, $x_2 \in V(H_{i-1}) \setminus \{u,w\} \subseteq A_2^e$, and $y_2, z_2 \in A_2^e$.
        However, this implies that $x_2$ is adjacent to $y_1$ and $z_1$ on $H_i$, which contradicts that no vertex in $V(H_{i-1}) \setminus \{u,w\}$ is adjacent to any vertex in $V(F_i)$ on $H_i$.

        If there is no vertex in $V(M)$ is in $V(H_{i-1}) \setminus \{u,w\}$, then there is an induced matching $M^\prime \subseteq M$ of $\indG{A_1^e, A_2^e}$ such that $V(M^\prime) \subseteq V(F_i)$ and $|M^\prime| \ge |M|-|V(M) \cap \{u,w\} | \ge \lmimwG{\overline{F}}+1$.
		For an edge $e^\prime$ of $T_i^\prime$ such that $A_1^{e^\prime} = A_1^e$ and $A_2^{e^\prime} = A_2^e \setminus V(H_{i-1})$, $M^\prime$ is also an induced matching of $\indG{A_1^{e^\prime} A_2^{e^\prime}}$ defined by the linear branch decomposition $(T^\prime_i, \bij^\prime_i)$.
		This implies that $\mimwT{T^\prime_i}{\bij^\prime_i} \ge \lmimwG{\overline{F}}+1$, which contradicts that $\mimwT{T^\prime_i}{\bij^\prime_i} \le \lmimwG{\overline{F}}$.
	\end{proof}

  \subsection{Special cases} \label{sec:NP-hardness_special}

    We strength \Cref{the:NP-hard_main} for several problems.
	We first give an improved upper bound on mim-width such that \textsc{Induced $\Prop$ Subgraph} and \textsc{$\Prop$ Vertex Deletion} are NP-hard.

	\begin{theorem}\label{the:NP-hard_mim}
		Let $\Prop$ be a fixed nontrivial hereditary graph property that admits all cliques.
        We denote by $F$ the base of $\compProp$-forbidden subgraphs.
        Then \textsc{Induced $\Prop$ Subgraph} and \textsc{$\Prop$ Vertex Deletion}, as well as their connected variants and their dominating variants, are NP-hard for graphs with mim-width at most $\max \{2, \mimwG{\overline{F}}+1\}$, even if a branch decomposition with mim-width at most $\max \{2, \mimwG{\overline{F}}+1\}$ of an input graph is given.
	\end{theorem}

	\begin{proof}
        We perform the same reduction as the proof of \Cref{the:NP-hard_main}.
        Recall that $\mathcal{H} = \langle H^\ast = H_0, H_1, \ldots, H_{n+m}=H^\prime \rangle$ is the sequence of subgraphs of $H^\prime$ such that $H_{i}$ for $i \in \segsingle{n}$ is obtained from $H_{i-1}$ by attaching a copy of $F^\prime$ to $v_i$, and $H_i$ for $i \in \seg{n+1}{n+m}$ is obtained from $H_{i-1}$ by attaching a copy of $F_{1,1}$ to $e_{i-n}$.
		We modify the proof of \Cref{lem:girth_mimw} and show that $\mimwG{\overline{H_i}} \le \max \{2, \mimwG{\overline{F}}+1\}$ for any $H_i \in \mathcal{H}$ by induction on $i$.

		First, instead of the concatenation of two linear branch decompositions, we define a \emph{graft} of two rooted layouts.
		Let $G_1$ and $G_2$ be vertex-disjoint induced subgraphs of a graph $G$ such that $V(G_1) \cup V(G_2) = V(G)$, and let $(T_1,\bij_1)$ and $(T_2,\bij_2)$ be rooted layouts of $G_1$ and $G_2$, where $T_1$ and $T_2$ have the roots $r_1$ and $r_2$, respectively.
		A graft of $(T_1,\bij_1)$ and $(T_2,\bij_2)$ is to construct a new rooted layout $(T,\bij)$ of $G$ as follows.
		A tree $T$ is obtained by connecting $r_1$ of $T_1$ and $r_2$ of $T_2$ with an edge $e$, and then inserting the root $r$ of $T$ into $e$.
		We set a bijection $\bij$ from $V(G)$ to the leaves of $T$ such that $\bij(v) = \bij_1(v)$ if $v \in V(G_1)$ and $\bij(v) = \bij_2(v)$ if $v \in V(G_2)$.

		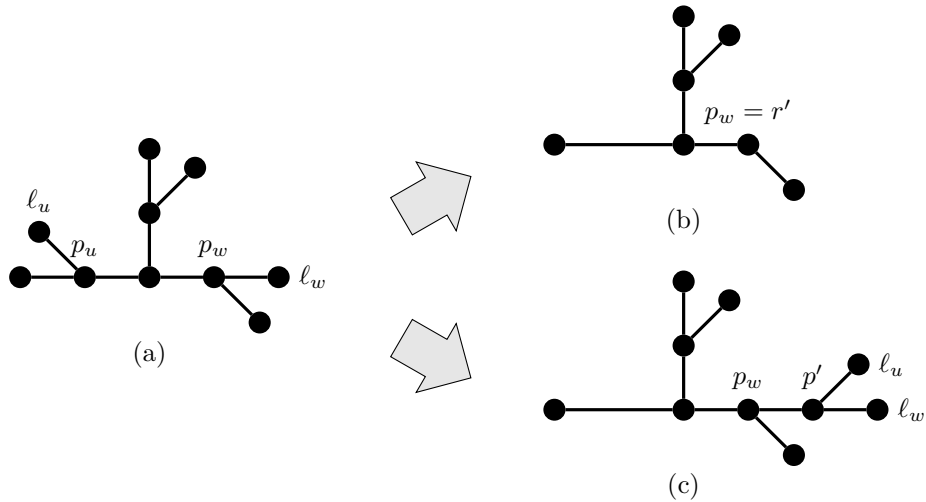
\begin{figure}[t]
			\centering
			\begin{tikzpicture}
				\def\edgew{0.85}
                \def\nodew{2.5}

				\node[draw, circle, very thick, label=90:{$p_u$}, fill = black, inner sep=\nodew] at (0,0) (t1){};
				\node[draw, circle, very thick, fill = black, inner sep=\nodew] at (0:\edgew) (t2){};
				\node[draw, circle, very thick, label=90:{$\ell_u$}, fill = black, inner sep=\nodew] at (135:\edgew) (t3){};
				\node[draw, circle, very thick, fill = black, inner sep=\nodew] at (180:\edgew) (t4){};
				\node[draw, circle, very thick, fill = black, inner sep=\nodew] at ($(t2) + (90:\edgew)$) (t5){};
				\node[draw, circle, very thick, fill = black, inner sep=\nodew] at ($(t5) + (90:\edgew)$) (t6){};
				\node[draw, circle, very thick, fill = black, inner sep=\nodew] at ($(t5) + (45:\edgew)$) (t7){};
				\node[draw, circle, very thick, label=90:{$p_w$}, fill = black, inner sep=\nodew] at ($(t2) + (0:\edgew)$) (t8){};
				\node[draw, circle, very thick, label=0:{$\ell_w$}, fill = black, inner sep=\nodew] at ($(t8) + (0:\edgew)$) (t9){};
				\node[draw, circle, very thick, fill = black, inner sep=\nodew] at ($(t8) + (-45:\edgew)$) (t10){};

				\draw [very thick] (t1) -- (t2);
				\draw [very thick] (t3) -- (t1) -- (t4);
				\draw [very thick] (t5) -- (t2) -- (t8);
				\draw [very thick] (t7) -- (t5) -- (t6);
				\draw [very thick] (t9) -- (t8) -- (t10);

				\node[circle, very thick] at ($(t2) + (-90:1.2*\edgew)$) (a){(a)};

				\node at (4.5, 1)[single arrow,draw=black,fill=black!10,minimum height=3em, minimum width=3em, single arrow head extend=0.3em, rotate = 30] (arrow){};

				\begin{scope}[xshift=20em, yshift=5em]
					\node[draw, circle, very thick, fill = black, inner sep=\nodew] at (0:\edgew) (t2){};
					\node[draw, circle, very thick, fill = black, inner sep=\nodew] at (180:\edgew) (t4){};
					\node[draw, circle, very thick, fill = black, inner sep=\nodew] at ($(t2) + (90:\edgew)$) (t5){};
					\node[draw, circle, very thick, fill = black, inner sep=\nodew] at ($(t5) + (90:\edgew)$) (t6){};
					\node[draw, circle, very thick, fill = black, inner sep=\nodew] at ($(t5) + (45:\edgew)$) (t7){};
					\node[draw, circle, very thick, label=90:{$p_w = r^\prime$}, fill = black, inner sep=\nodew] at ($(t2) + (0:\edgew)$) (t8){};
					\node[draw, circle, very thick, fill = black, inner sep=\nodew] at ($(t8) + (-45:\edgew)$) (t10){};

					\draw [very thick] (t2) -- (t4);
					\draw [very thick] (t5) -- (t2) -- (t8);
					\draw [very thick] (t6) -- (t5) -- (t7);
					\draw [very thick] (t8) -- (t10);

					\node[circle, very thick] at ($(t2) + (-90:1.2*\edgew)$) (b){(b)};
				\end{scope}

				\node at (4.5, -1)[single arrow,draw=black,fill=black!10,minimum height=3em, minimum width=3em, single arrow head extend=0.3em, rotate = -30] (arrow){};

				\begin{scope}[xshift=20em, yshift=-5em]
					\node[draw, circle, very thick, fill = black, inner sep=\nodew] at (0:\edgew) (t2){};
					\node[draw, circle, very thick, fill = black, inner sep=\nodew] at (180:\edgew) (t4){};
					\node[draw, circle, very thick, fill = black, inner sep=\nodew] at ($(t2) + (90:\edgew)$) (t5){};
					\node[draw, circle, very thick, fill = black, inner sep=\nodew] at ($(t5) + (90:\edgew)$) (t6){};
					\node[draw, circle, very thick, fill = black, inner sep=\nodew] at ($(t5) + (45:\edgew)$) (t7){};
					\node[draw, circle, very thick, label=90:{$p_w$}, fill = black, inner sep=\nodew] at ($(t2) + (0:\edgew)$) (t8){};
					\node[draw, circle, very thick, label=90:{$p^\prime$}, fill = black, inner sep=\nodew] at ($(t8) + (0:\edgew)$) (t9){};
					\node[draw, circle, very thick, fill = black, inner sep=\nodew] at ($(t8) + (-45:\edgew)$) (t10){};
					\node[draw, circle, very thick, label=0:{$\ell_w$}, fill = black, inner sep=\nodew] at ($(t9) + (0:\edgew)$) (t11){};
					\node[draw, circle, very thick, label=0:{$\ell_u$}, fill = black, inner sep=\nodew] at ($(t9) + (45:\edgew)$) (t12){};

					\draw [very thick] (t2) -- (t4);
					\draw [very thick] (t5) -- (t2) -- (t8);
					\draw [very thick] (t6) -- (t5) -- (t7);
					\draw [very thick] (t9) -- (t8) -- (t10);
					\draw [very thick] (t11) -- (t9) -- (t12);

					\node[circle, very thick] at ($(t2) + (-90:1.2*\edgew)$) (c){(c)};
				\end{scope}
			\end{tikzpicture}
			\caption{The examples of (a) a branch decomposition $(T^+_i,\bij^+_i)$ of $\overline{F_i^+}$, (b) a rooted layout $(T^{\prime}_i,\bij^{\prime}_i)$ of $\overline{F_i}$, and (c) a branch decomposition $(T^\ast_i,\bij^\ast_i)$ of $\overline{F_i^+}$, where $i \in \seg{n+1}{n+m}$ and $p_u \neq p_w$.}
			\label{fig:graft}
		\end{figure}

        Recall that $F_i$ for $i \in \segsingle{n+m}$ is the subgraph of $H_i$ obtained by removing all vertices in $V(H_{i-1})$.
        Analogous to the proof of \Cref{lem:girth_mimw}, we will construct a rooted layout $(T_i, \bij_i)$ of $\ovl{H_i}$ by a graft of rooted layouts of $\ovl{H_{i-1}}$ and $\ovl{F_i}$.
		To obtain a better bound of mim-width, we define a suitable rooted layout $(T^{\prime}_i,\bij^{\prime}_i)$ of $\ovl{F_i}$ as follows.

		When $i \in \segsingle{n}$, let $F_i^+$ be the subgraph of $H_i$ corresponding to a copy of $F^\prime$ attached to $v_i \in V(H)$.
        In other words, $F_i^+$ is the subgraph induced by $V(F_i) \cup \{ v_i \}$.
        (Assume that $|V(F_i^+)| \ge 2$ because $|V(F_i^+)| = 1$ is a trivial case where $V(F_i) = \emptyset$.)
		Consider an arbitrary branch decomposition $(T^+_i,\bij^+_i)$ of $\overline{F_i^+}$ such that $\mimwT{T^+_i}{\bij^+_i} \le \mimwG{\overline{F}}$.
		We denote by $\ell_{v_i}$ the leaf of $T^{+}_i$ with $\bij^{+}_i(v_i) = \ell_{v_i}$.
		The rooted full binary tree $T_i^\prime$ is constructed by removing $\ell_{v_i}$ from $T^{+}_i$ and set the neighbor of $\ell_{v_i}$ as the root $r^\prime$ of $T_i^\prime$.

		When $i \in \seg{n+1}{n+m}$, let $F_i^+$ be the subgraph of $H_i$ corresponding to a copy of $F_{1,1}$ attached to $e_{i-n} = uw \in E(H)$.
        In other words, $F_i^+$ is the subgraph induced by $V(F_i) \cup \{ u,w \}$.
  	(Assume that $|V(F_i^+)| \ge 3$ for the same reason as above.)
        Consider an arbitrary branch decomposition $(T^+_i,\bij^+_i)$ of $\overline{F_i^+}$ such that $\mimwT{T^+_i}{\bij^+_i} \le \mimwG{\overline{F}}$.
		We denote by $\ell_u$ and $\ell_w$ the leaves of $T^{+}_i$ with $\bij^{+}_i(u) = \ell_u$ and $\bij^{+}_i(w) = \ell_w$, respectively.
        In addition, denote by $p_u$ and $p_w$ the nodes of $T^{+}_i$ adjacent to $\ell_u$ and $\ell_w$, respectively.
        For a vertex $v$ of a graph with degree exactly~$2$ whose neighbors are two vertices $u$ and $w$, \emph{smoothing} $v$ is the process of removing $v$ and adding a new edge $uw$.
		If $p_u \neq p_w$, then the rooted full binary tree $T^{\prime}_i$ is constructed from $T^+_i$ by removing $\ell_u$ and $\ell_w$, smoothing $p_u$, and setting $p_w$ as the root $r^{\prime}$ of $T^{\prime}_i$.
		(See Figures~\ref{fig:graft}(a) and~(b).)
		If $p_u = p_w$, remove $\ell_u$, $\ell_v$, and $p_u$, and then set the neighbor of $p_u$ in $T^{+}_i$ other than $\ell_u$ and $\ell_v$ as the root $r^{\prime}$ of $T^{\prime}_i$.

		In any case where $i \in \segsingle{n+m}$, we let $L^{\prime}_i(x) = L^{+}_i(x)$ for every $x \in V(F_i)$.
		Finally, we define $(T_i, \bij_i)$ as a rooted layout obtained by a graft of $(T_{i-1},\bij_{i-1})$ and $(T^\prime_i,\bij^\prime_i)$.

		We prove that for every graph $H_i$ in $\mathcal{H}$, $\ovl{H_i}$ has mim-width at most $\max \{2, \mimwG{\overline{F}}+1\}$.
		We have shown that $\mimwG{\overline{H_0}} \le 2$ in the proof of \Cref{lem:girth_mimw}.
		Assume for a contradiction that there is an edge $e$ of $T_i$ such that the bipartite subgraph $\indG{A_1^e, A_2^e}$ of $\ovl{H_i}$ has an induced matching $M$ of size $\max \{3, \mimwG{\overline{F}}+2\}$.
		Then, we consider three cases: \ONE$^\prime$ $A_1^e \subseteq V(H_{i-1})$ and $V(F_i) \subseteq A_2^e$; \TWO$^\prime$ $A_1^e \subseteq V(F_i)$, $V(H_{i-1}) \subseteq A_2^e$ and $i \in \segsingle{n}$; and \THREE$^\prime$ $A_1^e \subseteq V(F_i)$, $V(H_{i-1}) \subseteq A_2^e$ and $i \in \seg{n+1}{n+m}$.
        In the same argument as the proof of \Cref{lem:girth_mimw}, contradictions can be derived for Cases~\ONE$^\prime$ and \TWO$^\prime$.
		We here discuss Case~\THREE$^\prime$.

		\smallskip
		\noindent \textbf{Case~\THREE$^\prime$:}
		Recall that at most two vertices in $V(H_{i-1})$, say $u$ and $w$, are adjacent to some vertex in $V(F_i)$ on $H_i$.
        As in Case~\TWO\ in the proof of \Cref{lem:girth_mimw}, we may assume that $V(M) \cap A_2^e \subseteq V(F_i) \cup \{u,w\}$.

		For the sake of contradiction, consider a branch decomposition $(T^\ast_i,\bij^\ast_i)$ of $\overline{F_i^+}$ constructed from $(T^+_i,\bij^+_i)$ as follows.
		Recall that $\ell_u$ and $\ell_w$ are the leaves of $T^{+}_i$ with $\bij^{+}_i(u) = \ell_u$ and $\bij^{+}_i(w) = \ell_w$, and $p_u$ and $p_w$ are the nodes of $T^{+}_i$ adjacent to $\ell_u$ and $\ell_w$, respectively.
		If $p_u \neq p_w$, then remove the edge $p_u \ell_u$ from $T^{+}_i$, smooth $p_u$, insert a new node $p^\prime$ into the edge $p_w \ell_w$, and then add an edge $p^\prime\ell_u$.
		(See Figures~\ref{fig:graft}(a) and~(c).)
		If $p_u = p_w$, then define $T^\ast_i = T^{+}_i$.
		This completes the construction of $T^\ast_i$.
		For each $x \in V(F_i^+)$, we let $\bij^\ast_i(x) = \bij^+_i(x)$.
		Observe that $T^\ast_i$ and $T^+_i$ differ in a location of at most one leaf.
		Thus, we have $\mimwT{T^\ast_i}{\bij^\ast_i} \le \mimwT{T^+_i}{\bij^+_i}+1 \le \mimwG{\overline{F}}+1$.
		Moreover, for any edge $e^\prime$ of $T^{\prime}_i$, there is an edge $e^\ast$ of $T^\ast_i$ such that $A_1^{e^\ast} = A_1^{e^\prime}$ and $A_2^{e^\ast}= A_2^{e^\prime} \cup \{u,w\}$.
		(See Figures~\ref{fig:graft}(b) and~(c).)

		Recall that $e$ is the edge of $T_i$ such that $A_1^e \subseteq V(F_i)$ and $V(H_{i-1}) \subseteq A_2^e$, which implies that $e$ is in the subtree $T_i^\prime$ of $T_i$.
		Thus, there is an edge $e^\prime$ of $T^\prime_i$ corresponding to $e$ such that $A_1^{e^\prime} = A_1^{e}$ and $A_2^{e^\prime}= A_2^{e} \cap V(F_i)$.
		Combined with $A_1^{e^\ast} = A_1^{e^\prime}$ and $A_2^{e^\ast}= A_2^{e^\prime} \cup \{u,w\}$, we have $A_1^{e^\ast} = A_1^e$ and $A_2^{e^\ast} = (A_2^{e} \cap V(F_i)) \cup \{u,w\}$.
		In addition, from the assumptions, $\indG{A_1^e, A_2^e}$ has an induced matching $M$ such that $V(M) \cap A_2^e \subseteq V(F_i) \cup \{u,w\}$ and $|M| \ge \mimwG{\overline{F}}+2 $.
		Therefore, $M$ is also an induced matching of $\indG{A_1^{e^\ast}, A_2^{e^\ast}}$ defined by the branch decomposition $(T^\ast_i, \bij^\ast_i)$, which contradicts that $\mimwT{T^\ast_i}{\bij^\ast_i} \le \mimwG{\overline{F}}+1$.
	\end{proof}

    As a consequence of \Cref{the:NP-hard_mim}, we have the following theorem.

	\begin{theorem}\label{the:NP-hard_mimtwo_H-free}
		Let $H$ be any fixed graph such that $\mimwG{H} \le 1$ and $H$ is not a complete graph.
		Then, \textsc{Induced $H$-free Subgraph} and \textsc{$H$-free Vertex Deletion}, as well as their connected variants and their dominating variants, are NP-hard for graphs with mim-width at most $2$, even if a branch decomposition with mim-width at most $2$ of an input graph is given.
	\end{theorem}

	\begin{proof}
		By \Cref{prop:mimone_comp}, $\mimwG{\ovl{H}} \le 1$ holds.
        Let $\compProp$ be a collection of $\ovl{H}$-free graphs.
		We show that the base $F$ of $\compProp$-forbidden subgraphs is chosen from induced subgraphs of $\ovl{H}$, which immediately implies \Cref{the:NP-hard_mimtwo_H-free} by combining \Cref{prop:mimw_subgraph,prop:mimone_comp} and \Cref{the:NP-hard_mim}.
		Assume for a contradiction that $F$ is not an induced subgraph of $\ovl{H}$.
		From the definition of $F$, there is an integer $\ell \ge 1$ such that $\ell F$ contains $\ovl{H}$ as an induced subgraph.
        If $\ell F = \ovl{H}$, then it immediately contradicts the assumption.
		Thus, $\ovl{H}$ is a proper induced subgraph of $F$.
        However, we have $\beta_{\ovl{H}}<_R \beta_F$, which contradicts the definition of $F$.
	\end{proof}

	\begin{figure}[t]
		\centering
		\begin{tabular}{ccc}

			\begin{minipage}[t]{0.3\linewidth}
				\centering
				\begin{tikzpicture}
                    \def\edgew{1}
                    \def\nodew{2.5}
					\node[draw, circle, very thick, label=180:{$v_2$}, fill = black, inner sep=\nodew] at (0,0) (v2){};
					\node[draw, circle, very thick, label=90:{$v_3$}, fill = black, inner sep=\nodew] at ($(v2) + (45:\edgew)$) (v3){};
					\node[draw, circle, very thick, label=-90:{$v_4$}, fill = black, inner sep=\nodew] at ($(v2) + (-45:\edgew)$) (v4){};
					\node[draw, circle, very thick, label=90:{$v_1$}, fill = black, inner sep=\nodew] at ($(v4) + (45:\edgew)$) (v1){};
					\node[draw, circle, very thick, label=90:{$v_5$}, fill = black, inner sep=\nodew] at ($(v1) + (45:\edgew)$) (v5){};
                    \node[draw, circle, very thick, label=-90:{$v_6$}, fill = black, inner sep=\nodew] at ($(v1) + (-45:\edgew)$) (v6){};

					\draw [very thick] (v2) -- (v3) -- (v1) -- (v2) -- (v4) -- (v1) -- (v5) -- (v6) -- (v1);
				\end{tikzpicture}

			\end{minipage} &




			\begin{minipage}[t]{0.3\linewidth}
				\centering
				\begin{tikzpicture}

                    \def\edgew{1}
                    \def\nodew{2.5}
					\node[draw, circle, very thick, label=180:{$v_2$}, fill = black, inner sep=\nodew] at (0,0) (v2){};
					\node[draw, circle, very thick, label=90:{$v_3$}, fill = black, inner sep=\nodew] at ($(v2) + (45:\edgew)$) (v3){};
					\node[draw, circle, very thick, label=-90:{$v_4$}, fill = black, inner sep=\nodew] at ($(v2) + (-45:\edgew)$) (v4){};
					\node[draw, circle, very thick, label=90:{$v_1$}, fill = black, inner sep=\nodew] at ($(v4) + (45:\edgew)$) (v1){};

					\draw [very thick] (v3) -- (v4);

				\end{tikzpicture}
			\end{minipage} &




			\begin{minipage}[t]{0.3\linewidth}
				\centering
				\begin{tikzpicture}
                    \def\edgew{0.75}
                    \def\nodew{2.5}
					\node[draw, circle, very thick, label=-90:{$v_1$}, fill = black, inner sep=\nodew] at (0,0) (v1){};
					\node[draw, circle, very thick, fill = black, inner sep=\nodew] at ($(v1) + (0:\edgew)$) (w2){};
					\node[draw, circle, very thick, fill = black, inner sep=\nodew] at ($(w2) + (0:\edgew)$) (w3){};
                    \node[draw, circle, very thick, label=-90:{$v_4$}, fill = black, inner sep=\nodew] at ($(w3) + (0:\edgew)$) (v4){};

                        \foreach \x in {2,3} {
                            \node[draw, circle, very thick, label=-90:{$v_\x$}, fill = black, inner sep=\nodew] at ($(w\x) + (-90:\edgew)$) (v\x){ };
                            \draw [very thick] (w\x) -- (v\x);
                        }

                        \draw [very thick] (v1) -- (w2) -- (w3) -- (v4);

				\end{tikzpicture}
			\end{minipage}
			\\
			(a) & (b) & (c)
		\end{tabular}
		\caption{ Let $\Prop$ be a collection of polar graphs. (a) The base $F$ of $\compProp$-forbidden subgraphs, where $F_{1,1}$ is the graph induced by $\{v_1,v_2,v_3,v_4\}$, $c_1 = v_1$, and $d = v_2$; (b) the complement graph $\ovl{F_{1,1}}$; and (c) the good linear branch decomposition of $\ovl{F_{1,1}}$.}
		\label{fig:base_polar}
	\end{figure}
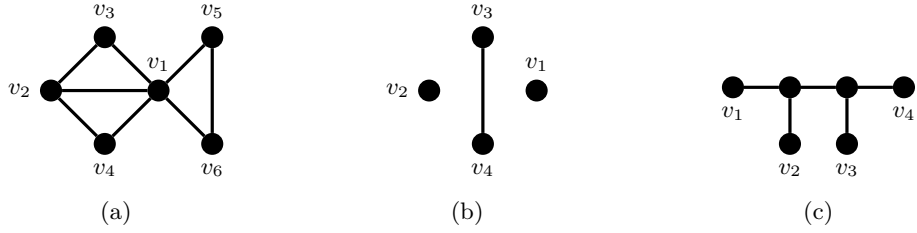

	Next, we give a stronger bound on linear mim-width.
    For the base $F$ of $\compProp$-forbidden subgraphs, recall that $c_1$ is the cut vertex of $F_1$ that realizes $\alpha_1$ and $d \in \Nei{F_{1,1}}{c_1}$ (see also \Cref{fig:forbidden_subgraph}).
    For a linear branch decomposition $(T,\bij)$ and a vertex $v$ of $\ovl{F_{1,1}}$, we denote by $p_v$ the adjacent node of $\bij(v)$.
    A linear branch decomposition $(T,\bij)$ of $\ovl{F_{1,1}}$ is called \emph{good} if $\mimwT{T}{\bij} \le \lmimwG{\ovl{F}}$ and $p_{c_1} = p_d$ (that is, one of them is an endpoint of the spine of $T$).

	\begin{theorem}\label{the:NP-hard_lmim}
		Let $\Prop$ be a fixed nontrivial hereditary graph property that admits all cliques.
        We denote by $F$ the base of $\compProp$-forbidden subgraphs.
        If there exists a good linear branch decomposition $(T,\bij)$ of $\ovl{F_{1,1}}$, then \textsc{Induced $\Prop$ Subgraph} and \textsc{$\Prop$ Vertex Deletion}, as well as their connected variants and their dominating variants, are NP-hard for graphs with linear mim-width at most $\max \{2, \lmimwG{\overline{F}}\}$, even if a branch decomposition with mim-width at most $\max \{2, \lmimwG{\overline{F}}\}$ of an input graph is given.
	\end{theorem}

	\begin{proof}
        This theorem is proved by modifying the proof of \Cref{the:NP-hard_mim}.
        The definition of the graph $H_i$ for $i \in \seg{0}{n+m}$ is the same.
        We modify the definitions of $(T_i, \bij_i)$, $(T^+_i,\bij^+_i)$, and $(T^\prime_i,\bij^\prime_i)$.
        Especially, we redefine them as linear branch decompositions.
        Let $(T_0, \bij_0)$ be an arbitrary linear branch decomposition of $\ovl{H_0}$.
        For $i \in \segsingle{n+m}$, we recursively construct a linear branch decomposition $(T_i, \bij_i)$ of $\ovl{H_i}$ as follows.

        When $i \in \segsingle{n}$, recall that $F_i^+$ is the subgraph of $H_i$ corresponding to a copy of $F^\prime$ attached to $v_i \in V(H)$.
        We define $(T^+_i,\bij^+_i)$ as an arbitrary linear branch decomposition of $\ovl{F_i^+}$ such that $\mimwT{T^+_i}{\bij^+_i} \le \lmimwG{\overline{F}}$.
        The subcubic tree $T_i^\prime$ is constructed by removing $\ell_{v_i}$ from $T^{+}_i$ and smoothing the neighbor $p_{v_i}$ of $\ell_{v_i}$.
        We denote by $t$ the node of $T_i^\prime$ that was adjacent to $p_{v_i}$ in $T^{+}_i$.

        When $i \in \seg{n+1}{n+m}$, recall that $F_i^+$ is the subgraph of $H_i$ corresponding to a copy of $F_{1,1}$ attached to $e_{i-n} = uw \in E(H)$, where $u$ and $v$ corresponding to $c_1$ and $d$ of $F_{1,1}$, respectively.
        Consider a good branch decomposition $(T^+_i,\bij^+_i)$ of $\ovl{F_i^+}$ such that $\mimwT{T^+_i}{\bij^+_i} \le \lmimwG{\ovl{F}}$.
        The subcubic tree $T_i^\prime$ is constructed from $T^+_i$ by removing $\ell_u$, $\ell_v$, and $p_u$, and then smoothing the neighbor $p$ of $p_u$ in $T^{+}_i$ other than $\ell_u$ and $\ell_v$.
        We denote by $t$ the node of $T_i^\prime$ that was adjacent to $p$ in $T^{+}_i$.

		In any case where $i \in \segsingle{n+m}$, we let $L^{\prime}_i(x) = L^{+}_i(x)$ for every $x \in V(F_i)$.
		Finally, we define $(T_i, \bij_i)$ as a linear decomposition obtained by a concatenation of $(T_{i-1},\bij_{i-1})$ and $(T^\prime_i,\bij^\prime_i)$, where the edge of $T^\prime_i$ incident to $t$ is operated for the concatenation.

        We prove that for every graph $H_i$ in $\mathcal{H}$, the complement $\ovl{H_i}$ has mim-width at most $\max \{2, \lmimwG{\ovl{F}}\}$.
        As in the proof of \Cref{the:NP-hard_mim}, assume for a contradiction that there is an edge $e$ of $T_i$ such that the bipartite subgraph $\indG{A_1^e, A_2^e}$ of $\ovl{H_i}$ has an induced matching $M$ of size $\max \{3, \lmimwG{\overline{F}}+1\}$.
        We again consider the three cases \ONE$^\prime$, \TWO$^\prime$, and \THREE$^\prime$ listed in the proof of \Cref{the:NP-hard_mim}, but we omit Cases~\ONE$^\prime$ and \TWO$^\prime$ because they can be shown in the same argument as the proof of \Cref{lem:girth_mimw}.

        \smallskip
		\noindent \textbf{Case~\THREE$^\prime$:} $A_1^e \subseteq V(F_i)$, $V(H_{i-1}) \subseteq A_2^e$ and $i \in \seg{n+1}{n+m}$.
		As in the proof of \Cref{the:NP-hard_mim}, we may assume that $V(M) \cap A_2^e \subseteq (A_2^e \cap V(F_i)) \cup \{u,w\}$.
		Observe that there is an edge $e^\prime$ of $T^\prime_i$ corresponding to $e$ such that $A_1^{e^\prime} = A_1^{e}$ and $A_2^{e^\prime}= A_2^{e} \cap V(F_i)$.
		Moreover, since $(T^\prime_i,\bij^\prime_i)$ is obtained from the good branch decomposition $(T^+_i,\bij^+_i)$, there is an edge $e^+$ of $T^+_i$ corresponding to $e^\prime$ such that $A_1^{e^+} = A_1^{e^\prime}$ and $A_2^{e^+}= A_2^{e^\prime} \cup \{u,v\}$, that is, $A_1^{e^+} = A_1^{e}$ and $A_2^{e^+}= (A_2^{e} \cap V(F_i)) \cup \{u,v\}$.
        This implies that $V(M) \cap A_1^e = V(M) \cap A_1^{e^+}$  and $V(M) \cap A_2^e \subseteq V(M) \cap A_2^{e^+}$.
		Therefore, $M$ is also an induced matching of $\indG{A_1^{e^+}, A_2^{e^+}}$ defined by the branch decomposition $(T^+_i, \bij^+_i)$, which contradicts that $\mimwT{T^+_i}{\bij^+_i} \le \lmimwG{\ovl{F}}$.
	\end{proof}

    We give two applications of \Cref{the:NP-hard_lmim}.
    A graph $G=(V,E)$ is a \emph{polar graph} if $V$ can be partitioned into two vertex sets $S_1$ and $S_2$ such that $\indG{S_1}$ is a cluster graph and $\indG{S_2}$ is a co-cluster graph, that is, the complement of a cluster graph.

    \begin{theorem} \label{the:NP-hard_mimtwo_polar}
		\textsc{Induced Polar Subgraph} and \textsc{Polar Vertex Deletion}, as well as their connected variants and their dominating variants, are NP-hard for graphs with linear mim-width at most~$2$, even if a branch decomposition with mim-width at most $2$ of an input graph is given.
	\end{theorem}

    \begin{proof}
        Let $\Prop$ be a collection of polar graphs.
        Observe that $\compProp$ is also a collection of polar graphs.
        We show that the base $F$ of $\compProp$-forbidden subgraphs is the graph depicted in \figurename~\ref{fig:base_polar}(a) and the linear branch decomposition $(T, \bij)$ of $\ovl{F_{1,1}}$ depicted in \figurename~\ref{fig:base_polar}(c) is good, where $\ovl{F_{1,1}}$ is depicted in \figurename~\ref{fig:base_polar}(b).
        Then \Cref{the:NP-hard_mimtwo_polar} follows from \Cref{the:NP-hard_mim}.

        If $F$ is indeed the base of $\compProp$-forbidden subgraphs, then $(T, \bij)$ is obviously good.
        To show that $F$ is the base of $\compProp$-forbidden subgraphs, we are required to verify that $F$ satisfies the following two conditions:
        \begin{enumerate}
    		\item there is an integer $\ell \ge 1$ such that $\ell F$ violates $\compProp$, whereas $(\ell-1) F$ satisfies $\compProp$; and
    		\item for any integer $\ell^\prime \ge 1$ and any graph $F^\prime$ with $\beta_{F^\prime} <_R \beta_{F}$, $\ell^\prime F^\prime$ satisfies $\compProp$.
    	\end{enumerate}
        By a brute-force search, one can see that $F$ is a polar graph and $2F$ is not a polar graph.
        We show that for any integer $\ell^\prime \ge 1$ and any graph $F^\prime$ with $\beta_{F^\prime} <_R \beta_{F} = \langle \langle 4, 3 \rangle \rangle$, $\ell^\prime F^\prime$ satisfies $\compProp$.
        Denote $\beta_{F^\prime} = \langle \alpha_1, \alpha_2, \ldots, \alpha_p \rangle$, where $p$ is the number of connected components of $F^\prime$.
        Recall that $\alpha_1 \ge_L \alpha_2 \ge_L \cdots \ge_L \alpha_p$, $H_i$ is a connected component of $F^\prime$ for each $i \in \segsingle{p}$, $C_{i,j}$ for $j \in \segsingle{q_i}$ is a connected component obtained by removing the cut vertex $c_i$ from $H_i$, $n_{i,j} = |V(C_{i,j})|+1$, and $\alpha_i = \langle n_{i,1}, n_{i,2}, \ldots, n_{i,q_i} \rangle$ is the non-increasing sequence.
        (See again the definitions in \Cref{sec:hardness_general}.)
        Since $\beta_{F^\prime} <_R \beta_{F} = \langle \langle 4, 3 \rangle \rangle$, we have $n_{i,j} \le 4$ for each $i \in \segsingle{p}$ and each $j \in \segsingle{q_i}$.
        For each $i \in \segsingle{p}$, we partition vertices of $H_i$ into a cluster set $A_i$ and an independent set $I_i$ as follows.

        Suppose that $n_{i,1} = 4$.
        Then, we have $n_{i,j} \le 2$ for each $j\in \seg{2}{q_i}$; otherwise, contradicting $\beta_{F^\prime} <_R \beta_{F}$.
        This implies that $C_{i,j} = K_1$ for each $j\in \seg{2}{q_i}$.
        Moreover, $C_{i,1}$ is isomorphic to $K_3$ or $P_3$.
        If $C_{i,1} = K_3$, let $I_i = \{c_i\}$ and $A_i = V(H_i) \setminus I_i$.
        If $C_{i,1} = P_3$, we denote by $w_i$ the vertex with degree $2$ of the subgraph induced by $V(C_{i,1})$, and let $A_i = \{c_i, w_i \}$ and $I_i = V(H_i) \setminus A_i$.
        Observe that $A_i$ is a cluster set and $I_i$ is an independent set of $H_i$.

        Suppose next that $n_{i,1} < 4$.
        Then, we have $n_{i,j} \le 3$ for each $j\in \segsingle{q_i}$.
        This implies that for each $j\in \segsingle{q_i}$, $C_{i,j}$ is isomorphic to $K_2$ or $K_1$.
        We let $I_i = \{c_i\}$ and $A_i = V(H_i) \setminus I_i$, and then clearly $A_i$ is a cluster set and $I_i$ is an independent set of $H_i$.

        Therefore, $\bigcup_{i \in \segsingle{p}} A_i$ is a cluster set of $F^\prime$ and $\bigcup_{i \in \segsingle{p}} I_i$ is an independent set of $F^\prime$, that is, $F^\prime$ is a polar graph.
        Moreover, for any integer $\ell^\prime \ge 1$, $\ell^\prime F^\prime$ is a polar graph.
        This completes the proof.
    \end{proof}

	\begin{theorem} \label{the:NP-hard_lmimtwo}
		All the following problems, as well as their connected variants and their dominating variants, are NP-hard for graphs with linear mim-width $2$:
		  (a) \textsc{Induced $H$-free Subgraph} for any fixed threshold graph $H$ that is not a complete graph; and
            (b) \textsc{Induced $\Prop$ Subgraph}, where $\Prop$ is a fixed nontrivial hereditary graph property that admits all cliques and excludes $\ovl{\ell K_2}$ for some integer $\ell \ge 1$.
      The NP-hardness for these problems holds even if a linear branch decomposition with mim-width at most~$2$ of an input graph is given.
	\end{theorem}

	\begin{proof}
		By \Cref{the:NP-hard_lmim}, it suffices to show that $\lmimwG{\ovl{F}} \le 2$ and there is a good linear branch decomposition of $\ovl{F_{1,1}}$.

        \smallskip
		\noindent \textbf{Case~(a):}
		As in the proof of \Cref{the:NP-hard_mimtwo_H-free}, $F$ is chosen from induced subgraphs of $\ovl{H}$.
        We claim that any linear branch decomposition $(T,\bij)$ of $\ovl{F}$ has mim-width at most~$1$.
        Assume for a contradiction that there exists an edge $e$ of $T$ such that the bipartite subgraph $\indG{A_1^e, A_2^e}$ of $\ovl{F}$ has an induced matching $\{x_1x_2, y_1y_2\}$, where $x_1, y_1 \in A_1^e$ and $x_2, y_2 \in A_2^e$.
		Observe that the vertex subset $\{x_1,x_2,y_1,y_2\}$ induces the subgraph of $F$ isomorphic to $2K_2$, $P_4$, or $C_4$.
        Thus, $\ovl{H}$ as well as $H$ contain $2K_2$, $P_4$, or $C_4$ as an induced subgraph, which is a contradiction because $2K_2$, $P_4$, and $C_4$ are forbidden subgraphs of threshold graphs~\cite{ChvatalH73}.
        Therefore, we have $\lmimwG{\ovl{F}} \le 1$.

        We show that there is a good linear branch decomposition of $\ovl{F_{1,1}}$.
        Suppose that $\lmimwG{\ovl{F}} = 0$, that is, $\ovl{F}$ has no edge.
        Then $F$ is a complete graph.
        Since $F$ has no cut vertex, we have $F_{1,1} = F$ and hence any linear branch decomposition of $\ovl{F_{1,1}}$ has mim-width~$0$.
        This immediately indicates that there is a good linear branch decomposition of $F_{1,1}$.
        Suppose next that $\lmimwG{\ovl{F}} = 1$.
        In fact, any linear branch decomposition $(T,\bij)$ of $\ovl{F}$ has mim-width exactly~$1$ by the above claim.
        Consider a linear branch decomposition $(T,\bij)$ of $\ovl{F}$ such that $p_{c_1} = p_d$, where $c_1$ and $d$ are the vertices of $F_{1,1}$.
        We define a linear branch decomposition $(T^\prime,\bij^\prime)$ of $\ovl{F_{1,1}}$ obtained from $T$ by removing nodes corresponding to the vertices in $V(F_{1,1}) \setminus V(F)$ and then smoothing each vertex of degree $2$.
        Then, $(T^\prime,\bij^\prime)$ is good because $\mimwT{T^\prime}{\bij^\prime} \le \mimwT{T}{\bij} = \lmimwG{\ovl{F}}$.

        \smallskip
		\noindent \textbf{Case~(b):}
        Let $\ell^* \ge 1$ be the minimum integer such that $\Prop$ contains no $\ovl{\ell^* K_2}$.
        Then, $\ell^* K_2$ is a forbidden subgraph of $\compProp$.
        Observe that the base $F$ of $\compProp$-forbidden subgraphs is $K_2$ and hence $F_{1,1} = F$.
        Furthermore, for any linear branch decomposition $(T,\bij)$ of $\ovl{F} = \ovl{F_{1,1}} = 2K_1$, we have $\mimwT{T}{\bij} = 0$.
        This immediately implies that there is a good linear branch decomposition of $\ovl{F_{1,1}}$.
	\end{proof}

    \textsc{Clique}, \textsc{Induced Cluster Subgraph}, \textsc{Induced $\ovl{P_3}$-free Subgraph}, and \textsc{Induced $\ovl{K_3}$-free Subgraph} are equivalent to \textsc{Induced $H$-free Subgraph} such that $H$ is $K_2$, $P_3$, $\ovl{P_3}$, and $K_3$, respectively, each of which is a threshold graph.
    \textsc{Induced Split Subgraph} is equivalent to \textsc{Induced $\Prop$ Subgraph} such that $\Prop$ is the collection of split graphs, which is nontrivial, hereditary, admits all cliques, and excludes $\ovl{2 K_2} = C_4$.
    Therefore, the following corollary is obtained from
    \Cref{the:NP-hard_mimtwo_polar,the:NP-hard_lmimtwo}.

    \begin{corollary} \label{cor:NP-hard_lmimtwo}
		All the following problems, as well as their connected variants and their dominating variants, are NP-hard for graphs with linear mim-width~$2$:
  	(i) \textsc{Clique};
		(ii) \textsc{Induced Cluster Subgraph};
		(iii) \textsc{Induced Polar Subgraph};
		(iv) \textsc{Induced $\ovl{P_3}$-free Subgraph};
		(v) \textsc{Induced Split Subgraph}; and
        (vi) \textsc{Induced $\ovl{K_3}$-free Subgraph}.
        \rev{The NP-hardness for these problems holds even if a linear branch decomposition with mim-width at most~$2$ of an input graph is given.}
	\end{corollary}

    \Cref{cor:NP-hard_lmimtwo} strongly suggests that the complements of graphs with linear mim-width~$2$ have unbounded mim-width, because \textsc{Independent Set}, the complementary problem of \textsc{Clique}, is solvable in polynomial time for bounded mim-width graphs.

	\section{Polynomial-time algorithms for graphs with mim-width at most 1} \label{sec:algo}

 \subsection{\textsc{Cluster Vertex Deletion}} \label{sec:cluster}

	A graph $G$ is called a \emph{cluster} if every connected component of $G$ is a complete graph.
	\textsc{Induced Cluster Subgraph} is equivalent to \textsc{Induced $P_3$-free Subgraph} and \textsc{Cluster Vertex Deletion} (in terms of polynomial-time solvability).
	From \Cref{cor:NP-hard_lmimtwo}, \textsc{Induced Cluster Subgraph} is NP-hard for graphs with linear mim-width at most~$2$.

	Recall that all graphs with mim-width at most~$1$ are perfect graphs by \Cref{prop:mimone_perfect}.
	It is known that \textsc{Clique} is solvable in polynomial time for perfect graphs~\cite{GrotschelLS88} and hence also for graphs with mim-width at most~$1$.
	In contrast, \textsc{Induced Cluster Subgraph} remains NP-hard for bipartite graphs~\cite{HsiehLLP24,Yannakakis81}, which are perfect graphs.
	Thus, the same argument as \textsc{Clique} is not applicable to \textsc{Induced Cluster Subgraph}.
	Nevertheless, assuming that a rooted layout $(T,\bij)$ of an input graph with $\mimwT{T}{\bij} = 1$ is given, we design a polynomial-time algorithm for \textsc{Induced Cluster Subgraph}.

	\begin{theorem}\label{the:cluster}
		Given a graph and its rooted layout of mim-width at most~$1$, \textsc{Induced Cluster Subgraph} is solvable in polynomial time.
	\end{theorem}

    It is known that all interval graphs, permutation graphs, distance-hereditary graphs, and convex graphs have mim-width at most~$1$ and their rooted layout of mim-width at most~$1$ can be obtained in polynomial time~\cite{BelmonteV13,Oum05}.
    Moreover, by \Cref{prop:mimone_comp}, rooted layouts with mim-width at most~$1$ for the complement of these graphs can also be obtained in polynomial time.
    Thus, our algorithm directly indicates the following corollary.

	\begin{corollary}\label{cor:cluster}
		There is an algorithm that solves \textsc{Induced Cluster Subgraph} in polynomial time for interval graphs, permutation graphs, distance-hereditary graphs, convex graphs, and their complements.
	\end{corollary}

    Here we give an idea of our algorithm.
    For a rooted layout $(T,\bij)$ of a given graph with mim-width at most~$1$, we compute an optimal solution by means of dynamic programming from the leaves to the root of $T$.
    To complete the computation in polynomial time, for each node $t$ of $T$, we discard redundant partial solutions and store essential ones of polynomial size.
    This approach was also employed in the previous algorithmic work of mim-width~\cite{BergougnouxDJ23,BergougnouxK21,BergougnouxPT22,BuixuanTV13,JaffkeKST19,JaffkeKT20a,JaffkeKT20b}.
    Especially, an equivalence relation called the \emph{$d$-neighbor equivalence} plays a key role in compressing partial solutions and designing XP algorithms parameterized by mim-width~\cite{BergougnouxDJ23,BergougnouxK21,BergougnouxPT22,BuixuanTV13,JaffkeKST19}.
    However, \Cref{cor:NP-hard_lmimtwo} suggests that the $d$-neighbor equivalence does not work for designing an algorithm for \textsc{Induced Cluster Subgraph}; otherwise, we would obtain an XP algorithm parameterized by mim-width, which is quite unlikely by \Cref{cor:NP-hard_lmimtwo}.
    A rooted layout with mim-width at most~$1$ resolves the difficulty.
    Recall that $\mimwT{T}{\bij} \le 1$ if and only if $\indG{A_1^e, A_2^e}$ for any edge $e$ of $T$ is a chain graph as in \Cref{prop:mimone_chain}.
	This property allows us to give strict total orderings of vertices in $A_1^e$ and $A_2^e$ with respect to neighbors of vertices.
    We define new equivalence relations over the strict total orderings, which enables the dynamic programming to run in polynomial time.

    Let $G$ be a graph and $\baseorder_A$ be a strict total order on $A \subseteq V(G)$.
    For a vertex subset $C \subseteq A$, we denote by $\heador{C}{\baseorder_A}$ and $\tailor{C}{\baseorder_A}$ the largest and smallest vertices in $C$ with respect to $\baseorder_A$, respectively.
    More precisely, for a vertex $u \in C$, $u = \heador{C}{\baseorder_A}$ if and only if $v \baseorder_A u$ for any vertex $v \in C \setminus \{u\}$, and $u = \tailor{C}{\baseorder_A}$ if and only if $u \baseorder_A w$ for any vertex $w \in C \setminus \{u\}$, respectively.
    (For the sake of convenience, we allow $C = \emptyset$ and in this case we let $\heador{C}{\baseorder_A} = \emptyset$ and $\tailor{C}{\baseorder_A} = \emptyset$.)
    For a subset $S \subseteq A$, a partition $(C_1, C_2, \ldots, C_p)$ of $S$ into $p$ disjoint subsets $C_1, C_2, \ldots, C_p$ is called a \emph{component partition of $S$ over $G$} if for every $i \in \segsingle{p}$, the subgraph of $G$ induced by $C_i$ is a connected component of $\indG{S}$.
    We say that a partition $(C_1, C_2, \ldots, C_p)$ of $S \subseteq A$ is \emph{indexed by} $\baseorder_A$ if it holds that $\heador{C_j}{\baseorder_A} \baseorder_A \heador{C_i}{\baseorder_A}$ for any pair of integers $i,j$ with $1 \le i < j \le p$.

    For a non-empty vertex subset $A$ of $G$ such that $\mimwe{A} \le 1$, a strict total order $\chainorder{A}$ of $A$ is called a \emph{chain order} if for any two distinct vertices $v, w$ in $A$, $v \chainorder{A} w$ means $\Nei{G}{v} \setminus A \subseteq \Nei{G}{w} \setminus A$.
    (If $|A| = 1$, we define that the trivial strict total order of $A$ is also a chain order.)
    By \Cref{prop:mimone_chain} and the definition of chain graphs, there is a chain order of $A$ if and only if $\mimwe{A} \le 1$.
    Note that $\mimwe{\ovl{A}} \le 1$ also holds and thus there is a chain order $\chainorder{\ovl{A}}$ of $\ovl{A}$.

    For subsets $S_A$ and $S_A^\prime$ of $A$, let $(C_1,C_2,\ldots, C_p)$ denote the component partition of $S_A$ over $G$ and let $(C_1^\prime,C_2^\prime,\ldots, C_q^\prime)$ denote the component partition of $S_A^\prime$ over $G$, where $p$ and $q$ are positive integers and both the component partitions are indexed by a chain order $\chainorder{A}$.
    We write $S_A \eqrelation{G,\chainorder{A}} S_{A}^\prime$ if $\heador{C_1}{\chainorder{A}} = \heador{C_1^\prime}{\chainorder{A}}$, $\tailor{C_1}{\chainorder{A}} = \tailor{C_1^\prime}{\chainorder{A}}$, and $\heador{C_2}{\chainorder{A}} = \heador{C_2^\prime}{\chainorder{A}}$.
    If the graph $G$ and the chain order $\chainorder{A}$ involved in the component partitions of $S_A$ and $S_A^\prime$ are clear from the context, then we use the shorthand $\eqrelation{A}$.
    It is not hard to see that $\eqrelation{A}$ is an equivalence relation over subsets of $A$.
    A \emph{representative} of $S_A$, denoted by $\rep{A}(S_A)$, is the set $R = \{ \heador{C_1}{\chainorder{A}}, \tailor{C_1}{\chainorder{A}}, \heador{C_2}{\chainorder{A}} \}$.
    In the same way, we define an equivalence relation $\eqrelation{\chainorder{\ovl{A}}}$ over subsets of $\ovl{A}$ according to a chain order $\chainorder{\ovl{A}}$ and a representative $\rep{\ovl{A}}(S_{\ovl{A}})$ of $S_{\ovl{A}} \subseteq \ovl{A}$.

    Consider two subsets $S_A, S_A^\prime \subseteq A$ with $|S_A| \ge |S_A^\prime|$.
    Assume that for any subset $S_{\ovl{A}} \subseteq \ovl{A}$, $S_A \cup S_{\ovl{A}}$ is a cluster set of $G$ if and only if $S_A^\prime \cup S_{\ovl{A}}$ is a cluster set of $G$.
    This suggests that there is no need to store $S_A^\prime$ during dynamic programming over $T$.
    Formally, we give the following lemma.

\begin{lemma}\label{lem:rep_equivalent}
    For a vertex subset $A$ of a graph $G$ such that $\mimwe{A} \le 1$, let $S_A, S_A^\prime \subseteq A$ be cluster sets of $G$ with $S_A \eqrelation{A} S_{A}^\prime$ and let $S_{\ovl{A}}$ be any subset of $\ovl{A}$.
    Then, $S_A \cup S_{\ovl{A}}$ is a cluster set of $G$ if and only if $S_A^\prime \cup S_{\ovl{A}}$ is a cluster set of $G$.
\end{lemma}

    \begin{proof}
        We may assume that $S_A$, $S_A^\prime$ and $S_{\ovl{A}}$ are all non-empty sets; otherwise, the lemma trivially holds.
        Suppose that $(X_1, \ldots, X_\alpha)$ is the component partition of $S_A$ indexed by $\chainorder{A}$; $(Y_1^\prime, \ldots, Y_\beta^\prime)$ is the component partition of $S_A^\prime$ indexed by $\chainorder{A}$; and $(Z_1, \ldots, Z_\gamma)$ is the component partition of $S_{\ovl{A}}$ indexed by $\chainorder{\ovl{A}}$, where $\alpha$, $\beta$, and $\gamma$ are positive integers.

        Suppose that $S_A \cup S_{\ovl{A}}$ is a cluster set of $G$.
        Note that $S_{\ovl{A}}$ is also a cluster set of $G$.
        For a positive integer $i$, we denote $x_i^{\mathsf{h}} = \heador{X_i}{\baseorder_A}$, $x_i^{\mathsf{t}} = \tailor{X_i}{\baseorder_A}$, $y_i^{\mathsf{h}} = \heador{Y_i}{\baseorder_A}$, $y_i^{\mathsf{t}} = \tailor{Y_i}{\baseorder_A}$, $z_i^{\mathsf{h}} = \heador{Z_i}{\chainorder{\ovl{A}}}$, and $z_i^{\mathsf{t}} = \tailor{Z_i}{\chainorder{\ovl{A}}}$.
        Recall that $x \chainorder{A} x_1^{\mathsf{h}}$ for any vertex $x \in S_A \setminus \{x_1^{\mathsf{h}}\}$, $y \chainorder{A} y_1^{\mathsf{h}}$ for any vertex $y \in S_A^\prime \setminus \{y_1^{\mathsf{h}}\}$, and $z \chainorder{\ovl{A}} z_1$ for any vertex $z \in S_{\ovl{A}} \setminus \{z_1\}$.

        If $x_1^{\mathsf{h}} z_1^{\mathsf{h}} \notin E(G)$, since $S_A \eqrelation{A} S_{A}^\prime$ and hence $x_1^{\mathsf{h}} = y_1^{\mathsf{h}}$, we have $y_1^{\mathsf{h}} z_1^{\mathsf{h}} \notin E(G)$.
        This implies that, as the bipartite subgraph $\indG{A, \ovl{A}}$ is a chain graph, there is no edge between each pair of $y \in S_A^\prime$ and $z \in S_{\ovl{A}}$, and hence $S_A^\prime \cup S_{\ovl{A}}$ is also a cluster set of $G$.

        If $x_1^{\mathsf{h}} z_1^{\mathsf{h}} \in E(G)$, then $X_1 \cup Z_1$ forms a maximal clique of $\indG{S_A \cup S_{\ovl{A}}}$.
        Thus, it holds that $x_1^{\mathsf{t}} z_1^{\mathsf{t}} \in E(G)$, $x_1^{\mathsf{h}} z_2^{\mathsf{h}} \notin E(G)$, and $x_2^{\mathsf{h}} z_1^{\mathsf{h}} \notin E(G)$.
        It follows from $S_A \eqrelation{t} S_{A}^\prime$ that we have $y_1^{\mathsf{t}} z_1^{\mathsf{t}} \in E(G)$, $y_1^{\mathsf{h}} z_2^{\mathsf{h}} \notin E(G)$, and $y_2^{\mathsf{h}} z_1^{\mathsf{h}} \notin E(G)$.
        This implies that there is an edge between any vertex in $Y_1$ and any vertex in $Z_1$; there is no edge between vertices $S_A^\prime \setminus Y_1$ and vertices in $S_{\ovl{A}}$; and there is no edge between vertices in $S_A^\prime$ and vertices in $S_{\ovl{A}} \setminus Z_1$.
        We conclude that $S_A^\prime \cup S_{\ovl{A}}$ is also a cluster set of $G$.
        Therefore, if $S_A \cup S_{\ovl{A}}$ is a cluster set of $G$, then $S_A^\prime \cup S_{\ovl{A}}$ is a cluster set of $G$, and vice versa.
    \end{proof}

    \Cref{lem:rep_equivalent} asserts that the equivalence relation $\eqrelation{A}$ allows us to determine vertex sets to be stored.
    However, \Cref{lem:rep_equivalent} is not enough to construct a dynamic programming algorithm.
    If a chain order is arbitrarily given for each node $t$ of $T$, then the ordering of the stored sets may change, which causes the algorithm to output an incorrect solution.
    To avoid the inconsistency, we need to define chain orders with additional constraints.

    Let $(T,\bij)$ be a rooted layout of a graph $G=(V,E)$.
	For a node $t$ of $T$, we denote by $T_t$ the subtree of $T$ rooted at $t$.
    We define $ V_t = \{\bij^{-1}(\ell) \mid  \text{$\ell$ is a leaf of $T_t$} \}$, $\overline{V_t} = V \setminus V_t$, $G_t = \indG{V_t}$, and
	$G_{\overline{t}} = \indG{\overline{V_t}}$.
	We use the shorthand notations $\nodecut{t}$ for the bipartite subgraph $\indG{V_t, \overline{V_t}}$ and $\rep{t}$ for the representative $\rep{V_t}$.
   We define a strict total order $\chainorder{t}$ on vertices in $V_t$, called a \emph{lower chain order}, that satisfies the two conditions below:
    \begin{enumerate} [($\ell$-1)]
		\item $\chainorder{t}$ is a chain order of $V_t$; and
		\item if $t$ has a child $c$, then for any pair of distinct vertices $v, w$ in $V_{c}$, it holds that $v \chainorder{c} w$ if and only if $v \chainorder{t} w$.
	\end{enumerate}
    \medskip
    We also define an \emph{upper chain order} $\chainorder{\ovl{t}}$ as a strict total order on vertices in $\ovl{V_t}$ that holds the following three conditions:
	\begin{enumerate} [($u$-1)]
		\item $\chainorder{\ovl{t}}$ is a chain order of $\ovl{V_t}$;
        \item if $t$ has a child $c$, then for any pair of distinct vertices $v, w$ in $\ovl{V_t}$, it holds that $v \chainorder{\ovl{t}} w$ if and only if $v \chainorder{\ovl{c}} w$; and
		\item if $t$ has the parent $p$, then for any pair of distinct vertices $v, w$ in $\ovl{V_t} \cap V_{p}$, it holds that $v \chainorder{\ovl{t}} w$ if and only if $v \chainorder{p} w$, where $\chainorder{p}$ is a lower chain order on $V_p$.
	\end{enumerate}

    \Cref{lem:chainorder} asserts that the above strict total orders can be found in polynomial time.
    The proof of \Cref{lem:chainorder} will be provided in \Cref{sec:chainorder}.

    \begin{lemma}\label{lem:chainorder}
    	   Let $(T,\bij)$ be a rooted layout of a graph $G$ with $\mimwT{T}{\bij} \le 1$.
            For every node $t$ of $T$, a lower chain order $\chainorder{t}$ and an upper chain order $\chainorder{\overline{t}}$ exist and can be obtained in polynomial time.
    \end{lemma}

    We here give the following two lemmas, which are keys to show the correctness of our algorithm given later.

    \begin{lemma}\label{lem:rep_child}
        Let $(T,\bij)$ be a rooted layout of a graph $G$ with $\mimwT{T}{\bij} \le 1$ and $t$ be an internal node of $T$ with a child $c$.
        For any subset $S \subseteq \ovl{V_{c}} \cap V_t$ of $G$, it holds that $\rep{\ovl{c}}(S) = \rep{t}(S)$.
    \end{lemma}

    \begin{proof}
        By the condition~($u$-3) for an upper chain order, for distinct vertices $v, w \in S$, it holds that $v \chainorder{\ovl{c}} w$ if and only if $v \chainorder{t} w$.
        Thus, the orderings of vertices in $S$ over $\chainorder{\ovl{c}}$ and $\chainorder{t}$ are the same, which immediately means that $\rep{\ovl{c}}(S) = \rep{t}(S)$.
    \end{proof}

    \begin{lemma}\label{lem:rep_composite}
        Let $(T,\bij)$ be a rooted layout of a graph $G$ with $\mimwT{T}{\bij} \le 1$ and let $t$ be an internal node of $T$ with children $a$ and $b$.
        For disjoint cluster sets $X \subseteq V_a$ and $Y \subseteq V_b$, if $X\cup Y$ is a cluster set of $G$, then $\rep{t}(X\cup Y) = \rep{t}( \rep{t}(X)\cup \rep{t}(Y))$ holds.
        Moreover, for a cluster set $Z \subseteq \ovl{V_t}$ of $G$, if $X \cup Z$ (resp.\ $Y \cup Z$) is a cluster set of $G$, then $\rep{\ovl{b}}(X \cup Z) = \rep{\ovl{b}}( \rep{\ovl{b}}(X)\cup \rep{\ovl{b}}(Z))$ (resp.\ $\rep{\ovl{a}}(Y \cup Z) = \rep{\ovl{a}}( \rep{\ovl{a}}(Y)\cup \rep{\ovl{a}}(Z))$) holds.
    \end{lemma}

        \begin{proof}
        We only prove that if $X\cup Y$ is also a cluster set of $G$, then $\rep{t}(X\cup Y) = \rep{t}( \rep{t}(X)\cup \rep{t}(Y))$, because the other claims can be shown in the same way.
       Suppose that $(X_1, \ldots, X_\alpha)$ is the component partition of $X$ and $(Y_1^\prime, \ldots, Y_\beta^\prime)$ is the component partition of $Y$, where $\alpha$ and $\beta$ are positive integers and both the component partitions are indexed by $\chainorder{A}$.

        We assume that $\alpha \ge 2$ and $\beta \ge 2$ because the lemma in the other cases can be proved in the same way.
        We denote $x_1^\textsf{h} = \head{t}{X_1}$, $x_1^\textsf{t} = \tail{t}{X_1}$, $x_2^\textsf{h} = \head{t}{X_2}$, $y_1^\textsf{h} = \head{t}{Y_1}$, $y_1^\textsf{t} = \tail{t}{Y_1}$, and $y_2^\textsf{h} = \head{t}{Y_2}$.
        Notice that $x_1^\textsf{t}$ (resp.\ $y_1^\textsf{t}$) may be identical to $x_1^\textsf{h}$ (resp.\ $y_1^\textsf{h}$).
        In addition, recall that $\rep{t}(X) = \{x_1^\textsf{h}, x_1^\textsf{t}, x_2^\textsf{h} \}$ and $\rep{t}(Y) = \{y_1^\textsf{h}, y_1^\textsf{t}, y_2^\textsf{h} \}$.
        To prove the lemma, we show that $X\cup Y \eqrelation{A} \rep{t}(X)\cup \rep{t}(Y)$.

        Let $(C_1, \ldots, C_p)$ be the component partition of $X\cup Y$ and let $(C_1^\prime, \ldots, C_q^\prime)$ be the component partition of $\rep{A}(X)\cup \rep{A}(Y)$, where $p$ and $q$ are positive integers and both the component partitions are indexed by $\chainorder{A}$.
        Our task is to show that $\head{t}{C_1} = \head{t}{C_1^\prime}$, $\tail{t}{C_1} = \tail{t}{C_1^\prime}$, and $\head{t}{C_2} = \head{t}{C_2^\prime}$.
        We consider the following two cases: \ONE~$x_1^\textsf{h}y_1^\textsf{h} \in E(G)$; and \TWO~$x_1^\textsf{h}y_1^\textsf{h} \notin E(G)$.

        \smallskip
		\noindent \textbf{Case~\ONE:}
        Since $X$, $Y$, and $X\cup Y$ are cluster sets of $G$, $X_1 \cup Y_1$ forms a clique of $G$.
        More precisely, $C_1 = X_1 \cup Y_1$ holds; if some vertex $z \in X \cup Y \setminus (X_1 \cup Y_1)$ is in $C_1$, then $z$ is adjacent to all vertices in $X_1 \cup Y_1$, which contradicts the component partitions of $X$ and $Y$.
        Similarly, $C_1^\prime = \{ x_1^\textsf{h}, y_1^\textsf{h}, x_1^\textsf{t}, y_1^\textsf{t} \}$ holds.
        Moreover, we have $\head{t}{C_1}, \head{t}{C_1^\prime} \in \{ x_1^\textsf{h}, y_1^\textsf{h} \}$ and $\head{t}{C_1}, \head{t}{C_1^\prime} \in \{ x_1^\textsf{t}, y_1^\textsf{t} \}$.
        If $\head{t}{C_1} = x_1^\textsf{h}$, then $\head{t}{C_1^\prime} = x_1^\textsf{h}$, and if $\head{t}{C_1} = y_1^\textsf{h}$, then $\head{t}{C_1^\prime} = y_1^\textsf{h}$.
        Thus, we have $\head{t}{C_1} = \head{t}{C_1^\prime}$, and similarly, we have $\tail{t}{C_1} = \tail{t}{C_1^\prime}$.

        We can also observe that $\head{t}{C_2} \in \{x_2^\textsf{h}, y_2^\textsf{h} \}$ and $\head{t}{C_2^\prime} \in \{x_2^\textsf{h}, y_2^\textsf{h} \}$.
        Without loss of generality, suppose that $\head{t}{C_2} = x_2^\textsf{h}$.
        This implies that $y_2^\textsf{h} \chainorder{A} x_2^\textsf{h}$.
        Therefore, we also have $\head{t}{C_2^\prime} = x_2^\textsf{h}$ as $C_1^\prime = \{ x_1^\textsf{h}, y_1^\textsf{h}, x_1^\textsf{t}, y_1^\textsf{t} \}$ and hence $\head{t}{C_2} = \head{t}{C_2^\prime}$.

        \smallskip
		\noindent \textbf{Case~\TWO:}
        In this case, we have $\head{t}{C_1} \in \{x_1^\textsf{h}, y_1^\textsf{h} \}$.
        Without loss of generality, suppose that $\head{t}{C_1}  = x_1^\textsf{h}$, that is, $X_1 \subseteq C_1$.
        We here claim that $C_1 = X_1$ holds.
        Assume for a contradiction that $C_1 \neq X_1$.
        Then, observe that $C_1 = X_1 \cup Y_i$ holds for some $i \in \seg{2}{\beta}$, and thus we have $\head{t}{X_1}\head{t}{Y_j} \in E(G)$.
        Recall that $a$ and $b$ are the children of $t$, $X \subseteq V_a$, and $Y \subseteq V_b$.
        This means that $X_1$ is the subset of $V_a$ and $Y_1, \ldots, Y_\beta$ are the subsets of $\ovl{V_{a}}$.
        By the condition ($\ell$-2) for the lower chain order and the condition ($u$-3) for the upper chain order, we have $\head{a}{X_1}\head{\ovl{a}}{Y_j} \in E(\nodecut{a})$.
        Moreover, since $\nodecut{a}$ is a chain graph, we have $\head{a}{X_1}\head{\ovl{a}}{Y_1} \in E(\nodecut{a})$.
        However, as $\head{a}{X_1} = x_1^\textsf{h}$ and $\head{\ovl{a}}{Y_1} = y_1^\textsf{h}$, this contradicts that $x_1^\textsf{h}y_1^\textsf{h} \notin E(G)$.

        By showing $C_1 = X_1$, we also have $C_1^\prime = \{ x_1^\textsf{h}, x_1^\textsf{t} \}$.
        This means that $\head{t}{C_1} = \head{t}{C_1^\prime}$ and $\tail{t}{C_1} = \tail{t}{C_1^\prime}$.
        In addition, we can see that $\head{t}{C_2} \in \{ y_1^\textsf{h}, x_2^\textsf{h}  \}$.
        If $\head{t}{C_2} = y_1^\textsf{h}$, then $z \chainorder{A} y_1^\textsf{h}$ holds for every vertex $z \in X\cup Y \setminus C_1$, and hence we have $\head{t}{C_2^\prime} = y_1^\textsf{h}$.
        Similarly, if $\head{t}{C_2} = x_2^\textsf{h}$, then $\head{t}{C_2^\prime} = x_2^\textsf{h}$.
        Therefore, $\head{t}{C_2} = \head{t}{C_2^\prime}$.
    \end{proof}

	We now provide a polynomial-time algorithm for \textsc{Induced Cluster Subgraph}.
	Suppose that $(T,\bij)$ is a rooted layout of a graph $G$ with $\mimwT{T}{\bij} \le 1$ and $t$ is a node of $T$.
    We let $\Repset{t} = \{\rep{t}(S_t) : S_t \subseteq V_t\}$ and $\Repset{\ovl{t}} = \{\rep{\ovl{t}}(S_{\ovl{t}}) : S_{\ovl{t}} \subseteq \ovl{V_t} \}$.
	For two sets $R_t \in \Repset{t}$ and $R_{\ovl{t}} \in \Repset{\ovl{t}}$, we define $\clf{t}{R_t}{R_{\ovl{t}}}$ as the function that returns the largest size of a subset $S_t \subseteq V_t$ such that
	\begin{enumerate}
  	\item $\rep{t}(S_t) = R_t$; and
		\item $S_t \cup R_{\ovl{t}}$ is a cluster set of $G$.
	\end{enumerate}
 \medskip
	We let $\clf{t}{R_{t}}{R_{\ovl{t}}} = -\infty$ if there is no subset satisfying the above conditions.
	For each triple of $t \in V(T)$, $R_t \in \Repset{t}$, and $R_{\ovl{t}} \in \Repset{\ovl{t}}$, we compute $\clf{t}{R_{t}}{R_{\ovl{t}}}$ by means of dynamic programming from the leaves to the root $r$ of $T$.
	As $G = G_r$, we obtain the maximum size of cluster sets of $G$ by computing $\min \{\clf{r}{R_{r}}{\emptyset} : R_{r} \in \Repset{r} \}$.
    Notice that, for simplicity, our algorithm computes the size of an optimal solution.
	One can easily modify our algorithm so that it finds the largest cluster set in the same time complexity.

	\smallskip
	\noindent \textbf{The case where $t$ is a leaf of $T$.}
	Denote by $v$ the unique vertex in $V_t$. Then, $\Repset{t} = \{ \emptyset, \{v\} \}$.
    If $R_t = \emptyset$, only $S_t = \emptyset$ satisfies the prescribed conditions for any $R_{\ovl{t}} \in \Repset{\ovl{t}}$.
    If $R_t = \{v\}$, then $S_t = \{v\}$ and we have to check that $\{v\} \cup R_{\ovl{t}}$ is a cluster set of $G$.
    In summary, we have
	\begin{align*}
		\clf{t}{R_{t}}{R_{\ovl{t}}} =
		\begin{cases}
			0 & \text{if $R_t = \emptyset$ and $R_{\ovl{t}}$ is a cluster set of $G$,} \\ 
			1 & \text{if $R_t = \{v\}$ and $\{v\} \cup R_{\ovl{t}}$ is a cluster set of $G$,} \\
			-\infty & \text{otherwise.}
		\end{cases}
	\end{align*}

	\smallskip
	\noindent \textbf{The case where $t$ is an internal node of $T$.}
	Suppose that $t$ has children $a$ and $b$, and $\clf{a}{R_{a}}{R_{\ovl{a}}}$ and $\clf{b}{R_{b}}{R_{\ovl{b}}}$ have already been computed for any $R_a \in \Repset{a}$, $R_{\ovl{a}} \in \Repset{\ovl{a}}$, $R_b \in \Repset{b}$, and $R_{\ovl{b}} \in \Repset{\ovl{b}}$.
    For the largest subset $S_t \subseteq V_t$ that satisfies the prescribed conditions, $S_t$ can be partitioned into two cluster sets $S_t \cap V_a$ and $S_t \cap V_b$.
    In addition, $(S_t \cap V_b)\cup R_{\ovl{t}}$ and $(S_t \cap V_a) \cup R_{\ovl{t}}$ form cluster sets of $\indG{V_{\ovl{a}}}$ and $\indG{V_{\ovl{b}}}$, respectively.
	We guess that $\rep{t}(S_t \cap V_a) = \rep{a}(S_t \cap V_a) = R_a \in \Repset{a}$ and $\rep{t}(S_t \cap V_b) = \rep{b}(S_t \cap V_b) = R_b \in \Repset{b}$.
    By \Cref{lem:rep_composite}, $R_t$ can be represented as follows:
    \begin{align*}
        R_t  & = \rep{t}(S_t) \\
        & = \rep{t}((S_t \cap V_a)\cup (S_t \cap V_b)) \\
        & = \rep{t}(\rep{t}(S_t \cap V_a)\cup \rep{t}(S_t \cap V_b)) \\
        & = \rep{t}(R_a \cup R_b).
	\end{align*}
    To obtain the value $\clf{t}{R_{t}}{R_{\ovl{t}}}$, we calculate the sum of $\clf{a}{R_{a}}{\rep{\ovl{a}}((S_t \cap V_b) \cup R_{\ovl{t}})}$ and $\clf{b}{R_{b}}{\rep{\ovl{b}}((S_t \cap V_a)\cup R_{\ovl{t}}) }$ for each pair $(R_a, R_b)$ such that $R_a \in \Repset{a}$, $R_b \in \Repset{b}$, and $R_t = \rep{t}(R_a \cup R_b)$.
    Combining \Cref{lem:rep_child,lem:rep_composite} with the fact that $\rep{\ovl{a}}(R_{\ovl{t}}) = \rep{\ovl{t}}(R_{\ovl{t}})$, which is observed from the condition ($u$-2) for an upper chain order, it holds that
	\begin{align*}
        \rep{\ovl{a}}((S_t \cap V_b) \cup R_{\ovl{t}}) & = \rep{\ovl{a}}(\rep{\ovl{a}}(S_t \cap V_b) \cup \rep{\ovl{a}}(R_{\ovl{t}})) \\
        & = \rep{\ovl{a}}(\rep{t}(S_t \cap V_b) \cup \rep{\ovl{t}}(R_{\ovl{t}})) \\
        & = \rep{\ovl{a}}(R_{b} \cup R_{\ovl{t}}).
	\end{align*}
    Similarly, we have $\rep{\ovl{b}}((S_t \cap V_a)\cup R_{\ovl{t}}) = \rep{\ovl{b}}(R_{a} \cup R_{\ovl{t}})$.
    We conclude that
	\begin{align*}
		\clf{t}{R_{t}}{R_{\ovl{t}}} = \max_{R_{a} \in \Repset{a} \land R_{b} \in \Repset{b} } \{ & \clf{a}{R_{a}}{\rep{\ovl{a}}(R_{b} \cup R_{\ovl{t}}) } \\
                                        & + \clf{b}{R_{b}}{\rep{\ovl{b}}(R_{a} \cup R_{\ovl{t}}) }  : R_t = \rep{t}(R_a \cup R_b)\}.
	\end{align*}

    We discuss the running time of our algorithm.
    Let $n$ be the number of vertices of a given graph $G$.
    Obviously, the total running time of our algorithm heavily depends on the computation of $\clf{t}{R_{t}}{R_{\ovl{t}}}$ for an internal node $t$ of $T$.
    We first construct $\Repset{t} = \{\rep{t}(S_t) : S_t \subseteq V_t\}$ and $\Repset{\ovl{t}} = \{\rep{\ovl{t}}(S_{\ovl{t}}) : S_{\ovl{t}} \subseteq \ovl{V_t} \}$.
    Notice that, from the definition of $\eqrelation{t}$, the sizes of $\rep{t}(S_t)$ and $\rep{\ovl{t}}(S_{\ovl{t}})$ are at most $3$ for any $S_t \subseteq V_t$ and any $S_{\ovl{t}} \subseteq \ovl{V_t}$, respectively.
    Thus, the construction can be done in $O(n^3)$ time.
    We then initialize $\clf{t}{R_{t}}{R_{\ovl{t}}} = -\infty$ for each pair of $R_t \in \Repset{t}$ and $R_{\ovl{t}} \in \Repset{\ovl{t}}$.
    This initialization takes $O(n^6)$ time.
    Finally, for each triple $R_a \in \Repset{a}$, $R_b \in \Repset{b}$, and $R_{\ovl{t}} \in \Repset{\ovl{t}}$, where $a$ and $b$ are children of $t$, we calculate $\clf{a}{R_{a}}{\rep{\ovl{a}}(R_{b} \cup R_{\ovl{t}}) } + \clf{b}{R_{b}}{\rep{\ovl{b}}(R_{a} \cup R_{\ovl{t}}) }$, and if the result is larger than the current stored value of $\clf{t}{R_{t}}{R_{\ovl{t}}}$ with $R_t = \rep{t}(R_a \cup R_b)$, update $\clf{t}{R_{t}}{R_{\ovl{t}}}$ to the calculated value.
    Notice that $\rep{\ovl{a}}(R_{b} \cup R_{\ovl{t}})$, $\rep{\ovl{b}}(R_{a} \cup R_{\ovl{t}})$, and $\rep{t}(R_a \cup R_b)$ are computed in $O(1)$ time because the sizes of $R_{b} \cup R_{\ovl{t}}$, $R_{a} \cup R_{\ovl{t}}$, and $R_a \cup R_b$ are at most $6$.
    This update is completed in $O(n^9)$ time.
    Since the tree $T$ has at most $O(n)$ nodes, our algorithm runs in $O(n^{10})$ time.
    This completes the proof of \Cref{the:cluster}.

    \subsection{Proof of \Cref{lem:chainorder}} \label{sec:chainorder}

    In this subsection, we show that \Cref{lem:chainorder}, whose proof was postponed in \Cref{sec:cluster}, by providing a polynomial-time algorithm that achieves our goal.

    \subsubsection{Algorithm} \label{sec:chainorder_alg}
        We first obtain a lower chain order $\chainorder{t}$ for every node $t$ of $T$ by constructing a directed graph $H_t^\ell$ with a bottom-up approach in $T$ as follows.
        If $t$ is a leaf of $T$, then $H_t^\ell$ is a graph with a single vertex $\bij^{-1}(t)$.
        For an internal node $t$ with children $a$ and $b$, suppose that $H_{a}^\ell$ and $H_{b}^\ell$ have already been constructed.
        For each of distinct vertices $u,w \in V_t$, a directed edge $(v,w)$ of $H_t^\ell$ is drawn if and only if one of the following conditions is satisfied: $(v,w)$ also exists in $H_{a}^\ell$; $(v,w)$ also exists in $H_{b}^\ell$; or $\Nei{\nodecut{t}}{v} \subset \Nei{\nodecut{t}}{w}$.
        We will explain later how $\chainorder{t}$ is obtained from $H_t^\ell$.

        After computing $\chainorder{t}$ for every node $t$ of $T$, we construct a directed graph $H_t^u$ by a top-down approach in $T$ to obtain an upper chain order $\chainorder{\overline{t}}$.
        For the root $r$ of $T$, $H_r^u$ is a graph with no vertices.
        For a node $t$ of $T$ with $t \neq r$, suppose that $p$ is the parent of $t$ and $H_{p}^u$ has already been constructed.
        The vertex set of $H_t^u$ equals to $\ovl{V_t}$.
        For each of distinct vertices $v,w \in \ovl{V_t}$, a directed edge $(v,w)$ of $H_t^u$ is drawn if and only if $(v,w)$ also exists in $H_{p}^u$ or $\Nei{\nodecut{t}}{v} \subset \Nei{\nodecut{t}}{w}$ holds.

        We will show later that for every node $t$ of $T$, $H_t^\ell$ and $H_t^u$ are directed acyclic graphs.
        Assuming for now that this is true, this leads us to define the \emph{layers} of $H_t^\ell$ and $H_t^u$.
        Let $L_0 = \emptyset$.
        For a positive integer $i$, we denote by $L_i$ the set of vertices with indegree zero of $H_t^\ell - \bigcup_{j \in \seg{0}{i-1} } L_{j}$.
        Note that there is an integer $\alpha \ge 1$ such that $V_t = \bigcup_{i \in \seg{0}{\alpha}} L_i$ because $H_t^\ell$ is a directed acyclic graph.
        Similarly, the layers of $H_t^u$ are defined.

        For each node $t$ of $T$, we construct a lower chain order $\chainorder{t}$ from the layers of $H_t^\ell$ as follows.
        Let $\originorder$ be an arbitrary strict total order of vertices in $G$, which is unchanged during constructing the lower chain orders of all nodes of $T$.
        For a positive integer $i$ such that $L_{i-1} = \emptyset$ and $L_i \neq \emptyset$, remove the smallest vertex from $L_i$ with respect to $\originorder$.
        Iterate this until $L_\alpha = \emptyset$, and define a lower chain order $\chainorder{t}$ so that $v \chainorder{t} w$ for distinct vertices $v,w$ of $H_t^\ell$ if and only if $v$ is removed from the layers of $H_t^\ell$ prior to $w$.
        We construct an upper chain order $\chainorder{\ovl{t}}$ in the same way, but remove the smallest vertex from $L_i$ with respect to the lower chain order $\chainorder{r}$ for the root $r$ of $T$.
       It is clear that $\chainorder{t}$ and $\chainorder{\ovl{t}}$ for each node $t$ of $T$ are computed in polynomial time.

        \subsubsection{Correctness of our algorithm} \label{sec:chainorder_correct}

        We here prove that for every node $t$ of $T$, $H_t^\ell$ and $H_t^u$ are directed acyclic graphs, which is postponed in \Cref{sec:chainorder_alg}.
        \Cref{lem:directed_edges} is an auxiliary claim to show the the correctness of our algorithm.

        \begin{lemma} \label{lem:directed_edges}
            Let $t$ be a node of $T$.
            For vertices $v,w$ of $V_t$, if there exists a directed edge $(v,w)$ in $H_t^\ell$, then $\Nei{\nodecut{t}}{v} \subseteq \Nei{\nodecut{t}}{w}$ holds.
            Similarly, for vertices $v,w$ of $\ovl{V_t}$, if there exists a directed edge $(v,w)$ in $H_t^u$, then $\Nei{\nodecut{t}}{v} \subseteq \Nei{\nodecut{t}}{w}$ holds.
        \end{lemma}
        \begin{proof}
            A directed edge $(v,w)$ in $H_t^\ell$ implies that there is a node $d$ of the subtree $T_t$ such that $\Nei{\nodecut{d}}{v} \subset \Nei{\nodecut{d}}{w}$.
            Since $\ovl{V_{t}} \subseteq \ovl{V_{d}}$, we have $\Nei{\nodecut{t}}{v} \subseteq \Nei{\nodecut{t}}{w}$.
            Similarly, a directed edge $(v,w)$ in $H_t^u$ implies that $\Nei{\nodecut{t}}{v} \subset \Nei{\nodecut{t}}{w}$ or there is an ancestor node $a$ of $t$ such that $\Nei{\nodecut{a}}{v} \subset \Nei{\nodecut{a}}{w}$.
            Since $V_{t} \subseteq V_{a}$, we conclude that $\Nei{\nodecut{t}}{v} \subseteq \Nei{\nodecut{t}}{w}$.
        \end{proof}

        \begin{lemma} \label{lem:DAG}
            For every node $t$ of $T$, $H_t^\ell$ and $H_t^u$ are directed acyclic graphs.
        \end{lemma}
        \begin{proof}
            We show that $H_t^\ell$ is a directed acyclic graph by induction on a node $t$ from leaves to the root of $T$.
            If $t$ is a leaf of $T$, then $H_t^\ell$ is clearly a directed acyclic graph.
            For an internal node $t$ of $T$ with children $a$ and $b$, assume for a contradiction that $H_t^\ell$ has a directed cycle $C = \langle c_1, c_2, \ldots, c_h, c_1 \rangle$ of length $h \ge 2$.
            By the induction hypothesis, $H_{a}^\ell$ and $H_{b}^\ell$ are directed acyclic graphs, and hence $C$ contains at least one directed edge $(v,w)$ such that $\Nei{\nodecut{t}}{v} \subset \Nei{\nodecut{t}}{w}$.
            Without loss of generality, suppose that $\Nei{\nodecut{t}}{c_h} \subset \Nei{\nodecut{t}}{c_1}$.
            On the other hand, a directed edge $(c_i, c_{i+1})$ for each $i \in \segsingle{h-1}$ suggests that $\Nei{\nodecut{t}}{c_i} \subseteq \Nei{\nodecut{t}}{c_{i+1}}$ from \Cref{lem:directed_edges}.
            Therefore, we have $\Nei{\nodecut{t}}{c_1} \subseteq \Nei{\nodecut{t}}{c_h}$, a contradiction.
            In the same way, it can be shown that $H_t^u$ is a directed acyclic graph.
        \end{proof}

        We next show that the strict total orders generated by our algorithm are indeed lower chain orders and upper chain orders.
        To this end, we give the following useful lemma.

        \begin{lemma} \label{lem:same_layer}
            For a node $t$ of $T$, let $v$ and $w$ be vertices that are in the same layer of $H_t^\ell$.
            Then, $\Nei{\nodecut{d}}{v} = \Nei{\nodecut{d}}{w}$ holds for every node $d \in V(T_t)$ such that $v,w \in V_d$.
            Moreover, if $t$ has the parent $p$, then $v$ and $w$ are also in the same layer of $H_{p}^\ell$.
        \end{lemma}

        \begin{proof}
            For the former claim, assume for a contradiction that there exists a node $d \in V(T_t)$ such that $v,w \in V_d$ and $\Nei{\nodecut{d}}{v} \neq \Nei{\nodecut{d}}{w}$.
            Then, the directed edge $(v,w)$ or $(w,v)$ is drawn in $H_t^\ell$ from the construction of $H_t^\ell$.
            This contradicts that $v$ and $w$ are in the same layer of $H_t^\ell$.

            We next show that the latter claim of the lemma.
            Assume for a contradiction that $v$ and $w$ are in different layers in $H_p^\ell$.
            Without loss of generality, suppose that $v \in L_j^p$ and $w \in L_k^p$ for $1 \le j < k \le \alpha_p$, where $L_i^p$ is the $i$-th layer of $H_p^\ell$ and $\alpha_p$ is the number of the layers of $H_p^\ell$.
            Then, for some vertex $x \in L_j^p$, there exists a directed path $P_{x,w}$ from $x$ to $w$; otherwise, there exists a vertex $z$ in the layer $L_{k^\prime}^p$ with $j < k^\prime \le k$ such that $z$ has indegree zero in $H_t^\ell - \bigcup_{j^\prime \in \seg{0}{j-1} } L_{j^\prime}^p$, which contradicts the definition of layers.

            Suppose that $v,w \in L_i$, where $L_i$ is the layer of $H_{t}^\ell$, $i\in \segsingle{\alpha}$, and $\alpha$ is the number of the layers of $H_{t}^\ell$.
            We lead a contradiction by induction on $i$.

            Consider the case where $i=1$.
            Then, there is no directed path from $x$ to $w$ in $H_t^\ell$.
            This implies that some directed edge of $P_{x,w}$ in $H_p^\ell$ is added when $H_p^\ell$ is constructed.
            Combined with \Cref{lem:directed_edges}, we have $\Nei{\nodecut{p}}{x} \subset \Nei{\nodecut{p}}{w}$.
            Since $v$ and $x$ are in the same layer of $H_p^\ell$, we also have $\Nei{\nodecut{p}}{v} = \Nei{\nodecut{p}}{x}$ from the former claim of \Cref{lem:same_layer} and hence $\Nei{\nodecut{p}}{v} \subset \Nei{\nodecut{p}}{w}$.
            On the other hand, $\Nei{\nodecut{t}}{v} = \Nei{\nodecut{t}}{w}$ holds because $v$ and $w$ are in the same layer of $H_t^\ell$.
            It follows from $\ovl{V_p} \subseteq \ovl{V_t}$ that we have $\Nei{\nodecut{p}}{v} = \Nei{\nodecut{p}}{w}$, a contradiction.

            Consider next the case where $i>1$ and assume that the latter claim of \Cref{lem:same_layer} holds for all pairs of two vertices in the layer $L_{i^\prime}$ of $H_t^\ell$ with $1 \le i^\prime < i$.
            If there is no directed path from $x$ to $w$ in $H_t^\ell$, then a contradiction occurs in the same reason above.
            Assume that there is a directed path from $x$ to $w$ in $H_t^\ell$, and let $L_{i^{\prime\prime}}$ be a layer of $H_t^\ell$ such that $x \in L_{i^{\prime\prime}}$ and $1 \le i^{\prime\prime} < i$.
            In addition, let $y$ be an arbitrary vertex in $L_{i^{\prime\prime}}$ such that there exists a directed path $P_{y,v}$ from $y$ to $v$ in $H_t^\ell$.
            Notice that the existence of $y$ is guaranteed in the same reason as that of $P_{x,w}$ in $H_p^\ell$.
            From the induction hypothesis, $x$ and $y$ are in the same layer in $H_p^\ell$.
            Therefore, all of $v$, $x$, and $y$ are in the same layer in $H_p^\ell$, which contradicts that there exists the directed path $P_{y,v}$ in $H_t^\ell$ and hence $P_{y,v}$ also exists in $H_p^\ell$ from the construction of $H_p^\ell$.
        \end{proof}

        For each node $t$ of $T$, one can have a lemma for the graph $H_t^u$ analogous to \Cref{lem:same_layer}.
        The next lemma completes the correctness of \Cref{lem:chainorder}.

        \begin{lemma}
            For each node $t$ of $T$, our algorithm generates a lower chain order $\chainorder{t}$ and an upper chain order $\chainorder{\ovl{t}}$.
        \end{lemma}
        \begin{proof}
            Consider the strict total order $\chainorder{t}$ generated by our algorithm.
            If $t$ is a leaf of $T$, then it is clear that $\chainorder{t}$ is a lower chain order.
            We assume that $t$ is an internal node of $T$.

            Let $v_1, v_2, \ldots, v_n$ be vertices in $V_t$, where $n = |V_t|$, that are removed from $H_t^\ell$ in this order by our algorithm.
            Recall that $v_i \chainorder{t} v_j$ holds for any $1 \le i < j \le n$.
            To show that $\chainorder{t}$ satisfies the condition ($\ell$-1), it suffices to claim that $\Nei{\nodecut{t}}{v_i} \subseteq \Nei{\nodecut{t}}{v_{i+1}}$ for each $i \in \segsingle{n-1}$.

            If $v_{i}$ and $v_{i+1}$ are in the same layer, we have $\Nei{\nodecut{t}}{v_i} = \Nei{\nodecut{t}}{v_{i+1}}$ from \Cref{lem:same_layer}.
            Consider the case where $v_{i}$ and $v_{i+1}$ are in different layers of $H_t^\ell$.
            If $v_{i}$ is in a layer $L_j$ for some $j \in \segsingle{\alpha-1}$, then $v_{i+1}$ is in a layer $L_{j+1}$.
            Moreover, $L_j$ has a vertex $x$ such that a directed edge $(x, v_{i+1})$ is drawn in $H_t^\ell$; otherwise, contradicting that $v_{i+1}$ is in $L_{j+1}$.
            From \Cref{lem:directed_edges}, the directed edge $(x, v_{i+1})$ implies that $\Nei{\nodecut{t}}{x} \subseteq \Nei{\nodecut{t}}{v_{i+1}}$.
            Thus, since $x$ and $v_i$ are in the same layer and hence $\Nei{\nodecut{t}}{v_i} = \Nei{\nodecut{t}}{x}$ holds from \Cref{lem:same_layer}, we have $\Nei{\nodecut{t}}{v_i} \subseteq \Nei{\nodecut{t}}{v_{i+1}}$.

            We next show that $\chainorder{t}$ satisfies the condition ($\ell$-2).
            Suppose that the node $t$ of $T$ has a child $c$.
            By symmetry, we only claim that for any distinct vertices $v,w \in V_c$, we have $v \chainorder{t} w$ if $v \chainorder{c} w$ holds.

            If $v$ and $w$ are in the same layer of $H_c^\ell$, then $v \chainorder{c} w$ means that $v \chainorder{0} w$.
            Moreover, $v$ and $w$ are also in the same layer of $H_t^\ell$ from \Cref{lem:same_layer}.
            Thus, we also have $v \chainorder{t} w$.
            Consider the case where $v$ and $w$ are in different layers of $H_c^\ell$.
            Analogous to the proof of \Cref{lem:same_layer}, there exists a vertex $x$ such that $x,v$ are in the same layer of $H_c^\ell$ and there exists a directed path $P_{x,w}$ from $x$ to $w$ in $H_c^\ell$.
            From \Cref{lem:same_layer}, this means that $x,v$ are in the same layer $L_i$ of $H_t^\ell$ for some $i \in \segsingle{\alpha-1}$.
            Moreover, since the directed path $P_{x,w}$ also exists in $H_t^\ell$ from the construction of $H_t^\ell$, $w$ is contained in a layer of $H_t^\ell$ after $L_i$.
            Therefore, we have $v \chainorder{t} w$.
            We conclude that the strict total order $\chainorder{t}$ generated by the above algorithm is indeed a lower chain order.

            In the same way as above, one can also verify that the strict total order $\chainorder{\ovl{t}}$ generated by our algorithm satisfies the conditions ($u$-1) and ($u$-2).
            We show that $\chainorder{\ovl{t}}$ satisfies the condition ($u$-3).
            Suppose that $t$ has the parent $p$.
            By symmetry, we only claim that for any distinct vertices $v,w \in V_p$, we have $v \chainorder{p} w$ if $v \chainorder{\ovl{t}} w$ holds.

            If $v$ and $w$ are in the same layer of $H_t^u$, then $v \chainorder{r} w$ holds from the construction of $\chainorder{\ovl{t}}$, where $r$ is the root of $T$.
            Moreover, since $v,w \in V_p$, we have $v,w \in V_{p^\prime}$ for every ancestor node $p^\prime$ of $t$.
            Applying the condition ($\ell$-2) for every node on the path from $r$ to $p$ iteratively, we have $v \chainorder{p} w$.
            Consider next the case where $v$ and $w$ are in different layers of $H_t^u$.
            Assume for a contradiction that $v \chainorder{\ovl{t}} w$ and $w \chainorder{p} v$ hold.
            Let $x$ be a vertex in $\ovl{V_t}$ such that $v$ and $x$ are in the same layer of $H_t^u$ and there exists a directed path $P_{x,w}$ from $x$ to $w$ in $H_t^u$.
            The existence of $x$ can be shown as in the proof of \Cref{lem:same_layer}.
            Since $w \in \ovl{V_t}$ and $w \notin \ovl{V_p}$, there is a vertex $w^\prime \in \ovl{V_t} \setminus \{w\}$ such that the directed edge $(w^\prime, w)$ in $P_{x,w}$ is newly added when $H_t^u$ is constructed.
            This implies that $\Nei{\nodecut{t}}{w^\prime} \subset \Nei{\nodecut{t}}{w}$.
            Applying \Cref{lem:directed_edges} along $P_{x,w}$, we also have $\Nei{\nodecut{t}}{x} \subset \Nei{\nodecut{t}}{w}$.
            In addition, since $v$ and $x$ are in the same layer of $H_t^u$, we have $\Nei{\nodecut{t}}{v} = \Nei{\nodecut{t}}{x}$ from \Cref{lem:same_layer} and hence $\Nei{\nodecut{t}}{v} \subset \Nei{\nodecut{t}}{w}$.
            It follows from $\Nei{\nodecut{t}}{v}, \Nei{\nodecut{t}}{w} \subseteq V_t$ that we obtain a vertex $y \in V_t$ such that $wy \in E(G)$ and $vy \notin E(G)$.
            On the other hand, consider a child $t^\prime$ of $p$ different from $t$.
            Observe that $v,w \in \ovl{V_t} \cap V_p$ means $v,w \in V_{t^\prime}$.
            Moreover, the condition ($\ell$-2) ensures that $w \chainorder{t^\prime} v$ from the assumption that $w \chainorder{p} v$.
            Thus, we have $\Nei{\nodecut{t^\prime}}{w} \subseteq \Nei{\nodecut{t^\prime}}{v}$ from the condition ($\ell$-1).
            However, since $y \in V_t$ means that $y \in \ovl{V_{t^\prime}}$, this contradicts that $wy \in E(G)$ and $vy \notin E(G)$.
        \end{proof}

    \subsection{\textsc{Induced Polar Subgraph}} \label{sec:polar}

    Recall that a graph $G=(V,E)$ is a polar graph if $V$ can be partitioned into two vertex sets $S_1$ and $S_2$ such that $\indG{S_1}$ is a cluster graph and $\indG{S_2}$ is a co-cluster graph.
    The \textsc{Polarity} problem asks whether a given graph is a polar graph.
    \textsc{Polarity} is known to be NP-complete for general graphs~\cite{ChernyakC86,Farrugia04}, and as far as we know, the complexity status of \textsc{Polarity} on perfect graphs is still open.
    Obviously, \textsc{Polarity} is solvable if \textsc{Induced Polar Subgraph} is also solvable.
    We here give a sketch of a polynomial-time algorithm for \textsc{Induced Polar Subgraph} that works on graphs with mim-width~$1$, including various subclasses of perfect graphs.

	\begin{theorem}\label{the:polarity}
		Given a graph and its rooted layout of mim-width at most $1$, \textsc{Induced Polar Subgraph} is solvable in polynomial time.
	\end{theorem}

    We adjust the algorithm for \textsc{Induced Cluster Subgraph} in \Cref{sec:cluster}.
    Let $G$ be a graph and $\baseorder_A$ be a strict total order on $A \subseteq V(G)$.
    A strict total order $\baseorder_A^\textsf{rev}$ is called a \emph{reverse order} of $\baseorder_A$ if for any distinct vertices $u,v \in A$, $u \baseorder_A w$ if and only if $w \baseorder_A^\textsf{rev} u$.

    For a non-empty vertex subset $A$ of $G$ such that $\mimwe{A} \le 1$, we denote by $\chainorder{A}^\textsf{rev}$ and $\chainorder{\ovl{A}}^\textsf{rev}$ the reverse orders of chain orders $\chainorder{A}$ and $\chainorder{\ovl{A}}$, respectively.
    We write $S_A \eqrelationcc{A} S_A^\prime$ if $S_A \eqrelation{\ovl{G}, \baseorder_{A}^\textsf{rev}} S_A^\prime$.
    In other words, for the component partition $(C_1,C_2,\ldots, C_p)$ of $S_A$ over $\ovl{G}$ and the component partition $(C_1^\prime,C_2^\prime,\ldots, C_q^\prime)$ of $S_A^\prime$ over $\ovl{G}$, where $p$ and $q$ are positive integers and both the component partitions are indexed by $\chainorder{A}^\textsf{rev}$, it holds that $\heador{C_1}{\chainorder{A}^\textsf{rev}} = \heador{C_1^\prime}{\chainorder{A}^\textsf{rev}}$, $\tailor{C_1}{\chainorder{A}^\textsf{rev}} = \tailor{C_1^\prime}{\chainorder{A}^\textsf{rev}}$, and $\heador{C_2}{\chainorder{A}^\textsf{rev}} = \heador{C_2^\prime}{\chainorder{A}^\textsf{rev}}$.

    We show the following lemma analogous to \Cref{lem:rep_equivalent}.
	\begin{lemma}\label{lem:repcc_equivalent}
		For a vertex subset $A$ of a graph $G$ such that $\mimwe{A} \le 1$, let $S_A, S_A^\prime \subseteq A$ be co-cluster sets of $G$ with $S_A \eqrelationcc{A} S_{A}^\prime$ and let $S_{\ovl{A}}$ be any subset of $\ovl{A}$.
		Then, $S_A \cup S_{\ovl{A}}$ is a co-cluster set of $G$ if and only if $S_A^\prime \cup S_{\ovl{A}}$ is a co-cluster set of $G$.
	\end{lemma}

    \begin{proof}
        Suppose that $S_A \cup S_{\ovl{A}}$ is a co-cluster set of $G$.
        Clearly, $S_A$, $S_A^\prime$, and $S_A \cup S_{\ovl{A}}$ are cluster sets of $\ovl{G}$.
        Recall that $S_A \eqrelationcc{A} S_{A}^\prime$ means that $S_A \eqrelation{\ovl{G}, \baseorder_{A}^\textsf{rev}} S_A^\prime$.
        By \Cref{lem:rep_equivalent}, $S_A^\prime \cup S_{\ovl{A}}$ is a cluster set of $\ovl{G}$, that is, $S_A^\prime \cup S_{\ovl{A}}$ is a co-cluster set of $G$.
        The converse is also true.
    \end{proof}

    We say that a pair $(S_{1,A},S_{2,A})$ of disjoint subsets of $A$ is a \emph{polar pair} if $S_{1,A}$ is a cluster set of $G$ and $S_{2,A}$ is a co-cluster set of $G$.
    For pairs $(S_{1,A},S_{2,A})$ and $(S_{1,A}^\prime,S_{2,A}^\prime)$ of disjoint subsets of $A$, we denote $(S_{1,A},S_{2,A}) \eqrelationp{A} (S_{1,A}^\prime, S_{2,A}^\prime)$ if $S_{1,A} \eqrelation{A} S_{1,A}^\prime$ and $S_{2,A} \eqrelationcc{A} S_{2,A}^\prime$.
    For a pair $(S_{1,A},S_{2,A})$ of disjoint subsets of $A$, a representative $\rep{t}(S_{1,A}, S_{2,A})$ is a pair $(R_{1,A},R_{2,A})$ such that the size of $R_{1,A} \cup R_{2,A}$ is minimized over all pairs with $(R_{1,A},R_{2,A}) \eqrelationp{t} (S_{1,A},S_{2,A})$.
    Similarly, we can define an equivalence relation $\eqrelationp{\ovl{A}}$ over subsets of $\ovl{A}$ and a representative $\rep{\ovl{A}}$ for a pair of disjoint subsets of $\ovl{A}$.
    Combined with \Cref{lem:rep_equivalent,lem:repcc_equivalent}, we can obtain the following proposition.
    \begin{proposition} \label{prop:rep_equivalent_polar}
        For a vertex subset $A$ of a graph $G$ such that $\mimwe{A} \le 1$, let $(S_{1,A},S_{2,A})$ and $(S_{1,A}^\prime,S_{2,A}^\prime)$ be polar pairs of $A$ with $(S_{1,A},S_{2,A}) \eqrelationp{A} (S_{1,A}^\prime,S_{2,A}^\prime)$ and let $(S_{1,\ovl{A}},S_{2,\ovl{A}})$ be any pair of disjoint subsets of $\ovl{A}$.
		Then,
        \begin{enumerate}
            \item $S_{1,A} \cup S_{1,\ovl{A}}$ is a cluster set of $G$ if and only if $S_{1,A}^\prime \cup S_{1,\ovl{A}}$ is a cluster set of $G$; and
            \item $S_{2,A} \cup S_{2,\ovl{A}}$ is a co-cluster set of $G$ if and only if $S_{2,A}^\prime \cup S_{2,\ovl{A}}$ is a co-cluster set of $G$.
        \end{enumerate}
    \end{proposition}

    Here, we describe a polynomial-time algorithm for \textsc{Induced Polar Subgraph}.
    We omit proofs for the correctness of our algorithm because they can be shown in the same way as in \Cref{sec:cluster}.
    Suppose that $(T,\bij)$ is a rooted layout of a graph $G$ with $\mimwT{T}{\bij} \le 1$ and $t$ is a node of $T$.
    Let $\Repset{t} = \{\rep{t}(S_{1,t},S_{2,t}) : S_{1,t},S_{2,t} \subseteq V_t, S_{1,t}\cap S_{2,t} = \emptyset \}$ and $\Repset{\ovl{t}} = \{\rep{\ovl{t}}(S_{1,\ovl{t}},S_{2,\ovl{t}}) : S_{1,\ovl{t}},S_{2,\ovl{t}} \subseteq \ovl{V_t}, S_{1,\ovl{t}}\cap S_{2,\ovl{t}} = \emptyset \}$.
	For two pairs $\mathcal{R}_t = (R_{1,t}, R_{2,t}) \in \Repset{t}$ and $\mathcal{R}_{\ovl{t}} = (R_{1,\ovl{t}}, R_{2,\ovl{t}}) \in \Repset{\ovl{t}}$, we define $\polf{t}{\mathcal{R}_t}{\mathcal{R}_{\ovl{t}}}$ as the function that returns the largest size of $S_{1,t} \cup S_{2,t}$ such that
	\begin{enumerate}
        \item $S_{1,t} \cap S_{2,t} = \emptyset$;
        \item $\rep{t}(S_{1,t}, S_{2,t}) = \mathcal{R}_t$; and
		\item $S_{1,t} \cup R_{1,\ovl{t}}$ is a cluster set of $G$ and $S_{2,t} \cup R_{2,\ovl{t}}$ is a co-cluster set of $G$.
	\end{enumerate}
 \medskip
	We let $\polf{t}{\mathcal{R}_t}{\mathcal{R}_{\ovl{t}}} = -\infty$ if there is no subset satisfying the above conditions.
	We compute $\polf{t}{\mathcal{R}_t}{\mathcal{R}_{\ovl{t}}}$ by means of dynamic programming from leaves to the root $r$ of $T$ as follow.
    If $t$ is a leaf of $T$, we have
    \begin{align*}
		\polf{t}{\mathcal{R}_t}{\mathcal{R}_{\ovl{t}}} =
		\begin{cases}
			    0 & \text{if $\mathcal{R}_t = (\emptyset, \emptyset)$, $R_{1,\ovl{t}}$ is a cluster set of $G$,} \\
              & \text{~~~and $R_{2,\ovl{t}}$ is a co-cluster set of $G$;} \\
			1 & \text{if $\mathcal{R}_t = (\{v\}, \emptyset)$, $\{v\} \cup R_{1,\ovl{t}}$ is a cluster set of $G$,} \\
              & \text{~~~and $R_{2,\ovl{t}}$ is a co-cluster set of $G$;} \\
            1 & \text{if $\mathcal{R}_t = (\emptyset, \{v\})$, $R_{1,\ovl{t}}$ is a cluster set of $G$,} \\
              & \text{~~~and $\{v\} \cup R_{2,\ovl{t}}$ is a co-cluster set of $G$;} \\
			-\infty & \text{otherwise,}
		\end{cases}
	\end{align*}
    where $v$ is the unique vertex in $V_t$.

    For two pairs $\mathcal{R} = (R_1, R_2)$ and $\mathcal{R}^\prime = (R_1^\prime, R_2^\prime)$, we denote $\mathcal{R} \uplus \mathcal{R}^\prime = (R_1\cup R_1^\prime, R_2 \cup R_2^\prime)$.
    If $t$ is an internal node of $T$ with children $a$ and $b$, we have
    \begin{align*}
		\polf{t}{\mathcal{R}_t}{\mathcal{R}_{\ovl{t}}} = \max_{\mathcal{R}_{a} \in \Repset{a} \land \mathcal{R}_{b} \in \Repset{b} } \{ & \clf{a}{\mathcal{R}_{a}}{\rep{\ovl{a}}(\mathcal{R}_{b} \uplus \mathcal{R}_{\ovl{t}}) } \\
        & + \clf{b}{\mathcal{R}_{b}}{\rep{\ovl{b}}(\mathcal{R}_{a} \uplus \mathcal{R}_{\ovl{t}}) }  : \mathcal{R}_t = \rep{t}(\mathcal{R}_a \uplus \mathcal{R}_b)\}.
	\end{align*}

    Each of $\Repset{a}$, $\Repset{b}$, and $\Repset{\ovl{t}}$ consists of vertex sets of size at most $6$ and $T$ has $O(n)$ nodes.
    Therefore, our algorithm for \textsc{Induced Polar Subgraph} runs in $O(n^{6\cdot3+1}) = O(n^{19})$ time.
    This completes the proof of \Cref{the:polarity}.

    \subsection{Connected variant} \label{sec:con}

    In this section, we briefly explain how to handle a connected variant of \textsc{Induced $\Prop$ Subgraph}.

    For vertex subsets $S_A, S_A^\prime$ of $A \subseteq V(G)$, let $(C_1,\ldots, C_{p})$ and $(C_1^\prime,\ldots, C_{q}^\prime)$ be the component partitions of $S_A$ and $S_A^\prime$ over $G$ indexed by a strict total order $\chainorder{A}$, respectively, where $p$ and $q$ are positive integers.
    We define an equivalence relation $\eqrelationc{A}$ as follows: $S_A \eqrelationc{A} S_A^\prime$ if $\head{A}{C_1} = \head{A}{C_1^\prime}$ and $\head{A}{C_p} = \head{A}{C_q^\prime}$.

    \begin{lemma}\label{lem:rep_equivalent_con}
		For a vertex subset $A$ of a graph $G$ such that $\mimwe{A} \le 1$, let $S_A$ and $S_A^\prime$ be vertex subsets of $A$ with $S_A \eqrelationc{A} S_A^\prime$ and let $S_{\ovl{A}}$ be any vertex subset of $\ovl{A}$.
		Then, $\indG{S_A \cup S_{\ovl{A}}}$ is connected if and only if $\indG{S_A^\prime \cup S_{\ovl{A}}}$ is connected.
	\end{lemma}

    \begin{proof}
		Suppose that $(X_1, \ldots, X_\alpha)$ is the component partition of $S_A$ indexed by $\chainorder{A}$; $(Y_1^\prime, \ldots, Y_\beta^\prime)$ is the component partition of $S_A^\prime$ indexed by $\chainorder{A}$; and $(Z_1, \ldots, Z_\gamma)$ is the component partition of $S_{\ovl{A}}$ indexed by $\chainorder{\ovl{A}}$, where $\alpha$, $\beta$, and $\gamma$ are positive integers.
        To prove the lemma, we consider the following two cases: \ONE\ $S_{\ovl{A}} = \emptyset$; and \TWO\ $S_{\ovl{A}} \neq \emptyset$.

        Suppose that $S_{\ovl{A}} = \emptyset$.
        If $\indG{S_A}$ is connected, then we have $p=1$.
        Moreover, as $S_A \eqrelationc{A} S_A^\prime$, $\head{A}{C_1} = \head{A}{C_1^\prime}$ and $\head{A}{C_p} = \head{A}{C_q^\prime}$ hold.
        Combined with $p=1$, this implies that $\head{A}{C_1^\prime} = \head{A}{C_q^\prime}$, that is, $\indG{S_A^\prime}$ consists of a single connected component.
        Therefore, when $S_{\ovl{A}} = \emptyset$, if $\indG{S_A}$ is connected, then $\indG{S_A^\prime}$ is also connected, and vice versa.

        Suppose next that $S_{\ovl{A}} \neq \emptyset$.
        By symmetry, we only show the sufficiency.
        For a positive integer $i$, we denote $x_i = \heador{X_i}{\baseorder_A}$, $y_i = \heador{Y_i}{\baseorder_A}$, and $z_i = \heador{Z_i}{\chainorder{\ovl{A}}}$.
        If $\indG{S_A \cup S_{\ovl{A}}}$ is connected, then we have $x_1 z_\gamma \in E(G)$; otherwise, since $\indG{A, \ovl{A}}$ is the chain graph, there is no edge between vertices in $S_{A}$ and vertices in $Z_\gamma$, which contradicts that $\indG{S_A \cup S_{\ovl{A}}}$ is connected.
        For the same reason, we have $x_\alpha z_1 \in E(G)$.
        From the definition of $\eqrelationc{A}$, this means that $y_\beta z_1 \in E(G)$.
        Any vertex in $S_A^\prime \setminus Y_\beta$ is adjacent to $z_1$ and there is a path from any vertex in $Y_\beta$ to $z_1$ via $y_\beta$.
        Similarly, any vertex in $S_{\ovl{A}} \setminus Z_\gamma$ is adjacent to $x_1$ and there is a path from any vertex in $Z_\gamma$ to $x_1$ via $z_\gamma$.
        Moreover, $x_1 z_1 \in E(G)$ because $x_1 z_\gamma \in E(G)$ and $\indG{A, \ovl{A}}$ is the chain graph.
        We conclude that there is a path between any two vertices in $S_A^\prime \cup S_{\ovl{A}}$, that is, $\indG{S_A^\prime \cup S_{\ovl{A}}}$ is connected.
    \end{proof}

    As an example, consider \textsc{Connected Induced Polar Subgraph}.
    Given a rooted layout $(T, \bij)$ of a graph $G$ with $\mimwT{T}{\bij} \le 1$, we define a new equivalence relation $\eqrelationcp{t}$ for each node $t$ of $T$ as follows: for pairs $(S_{1,t},S_{2,t})$ and $(S_{1,t}^\prime,S_{2,t}^\prime)$ of disjoint subsets of $V_t$, we write $(S_{1,t},S_{2,t}) \eqrelationcp{t} (S_{1,t}^\prime,S_{2,t}^\prime)$ if $S_{1,t}\cup S_{2,t}\eqrelationc{t} S_{1,t}^\prime\cup S_{2,t}^\prime$ and $(S_{1,t},S_{2,t}) \eqrelationp{t} (S_{1,t}^\prime,S_{2,t}^\prime)$.
    We also define $\eqrelationcp{\ovl{t}}$ analogously.
    By slightly modifying the algorithm in \Cref{sec:polar}, we can construct a polynomial-time algorithm for \textsc{Connected Induced Polar Subgraph} based on $\eqrelationcp{t}$ and $\eqrelationcp{\ovl{t}}$.

    \subsection{More general cases} \label{sec:meta_algo}
    In this section, we consider more general problems.
    Let $k$ be a fixed integer and let $\Prop_1, \Prop_2,\ldots, \Prop_k$ be $k$ graph properties.
    A $(\Prop_1, \Prop_2,\ldots, \Prop_k)$-\emph{coloring} of a graph $G=(V,E)$ is a partition $(V_1, V_2, \ldots, V_k)$ of $V$ such that $V_i$ is a $\Prop_i$-set of $G$ for every $i \in \segsingle{k}$.
    The $(\Prop_1, \Prop_2,\ldots, \Prop_k)$-\textsc{Partition} problem asks whether a given graph $G$ has a $(\Prop_1, \Prop_2,\ldots, \Prop_k)$-coloring.
    A graph property $\Prop$ is said to be \emph{additive} if for any two graphs $G_1$ and $G_2$ that satisfy $\Prop$, the disjoint union of $G_1$ and $G_2$ also satisfies $\Prop$.
    For any fixed nontrivial additive hereditary graph properties $\Prop_1, \Prop_2,\ldots, \Prop_k$, Farrugia showed that $(\Prop_1, \Prop_2,\ldots, \Prop_k)$-\textsc{Partition} is NP-hard unless $k=2$ and both $\Prop_1$ and $\Prop_2$ are the classes of edgeless graphs (that is, unless the problem of finding a proper $2$-coloring of $G$)~\cite{Farrugia04}.
    Furthermore, in~\cite{TamuraIZ20}, the minimization variant of $(\Prop_1, \Prop_2,\ldots, \Prop_k)$-\textsc{Partition}, namely \textsc{Min $(\Prop_1, \Prop_2,\ldots, \Prop_k)$-Partition}, was introduced, where we are required to find a $(\Prop_1, \Prop_2,\ldots, \Prop_k)$-coloring such that $V_1$ is minimized.
    Naturally, the maximization variant of $(\Prop_1, \Prop_2,\ldots, \Prop_k)$-\textsc{Partition}, namely \textsc{Max $(\Prop_1, \Prop_2,\ldots, \Prop_k)$-Partition}, can also be defined in the similar way.

    We here define the generalized problem including the above problems.
    A \emph{degree constraint matrix} $D_k$, which was introduced by Telle and Proskurowski~\cite{TelleP97}, is a $k \times k$ matrix such that each entry is a subset of the natural numbers $\mathbb{N}$.
    Let $\textsf{R} \in \{\textsc{Ext}, \textsc{Min}, \textsc{Max} \}$ be a type of a problem.
    The \textsf{R}-$(\Prop_1, \Prop_2,\ldots, \Prop_k)$-\textsc{Partition With $D_k$} asks for finding a partition $(V_1, V_2, \ldots, V_k)$ of a vertex set of a given graph $G$ that satisfies the following three conditions:
    \begin{itemize}
        \item $(V_1, V_2, \ldots, V_k)$ is a $(\Prop_1, \Prop_2,\ldots, \Prop_k)$-coloring of $G$;
        \item for each $i,j \in \segsingle{k}$ and each $v \in V_i$, it holds that $|\Nei{G}{v} \cap V_j | \in D_k[i,j]$; and
        \item $V_1$ is minimized if $\textsf{R} = \textsc{Min}$, $V_1$ is maximized if $\textsf{R} = \textsc{Max}$, and we do not care about the size of $V_1$ if $\textsf{R} = \textsc{Ext}$.
    \end{itemize}

    For instance, for the degree constraint matrix $D_2$ defined as follows, \textsc{Min}-$(\Prop_1, \Prop_2)$-\textsc{Partition With $D_2$} such that each of $\Prop_1$ and $\Prop_2$ contains all graphs is equivalent to \textsc{Dominating Set}:
    \[
    D_2 =
    \begin{pmatrix}
        \mathbb{N} & \mathbb{N} \\
        \mathbb{N} \setminus \{ 0 \}  & \mathbb{N} \\
    \end{pmatrix}.
    \]
    The entry $D_k[2,1]$ promises that each vertex in $V_2$ is dominated by at least one vertex in $V_1$.
    Furthermore, when $\Prop_1$ is the collection of all cliques and $\Prop_2$ contains all graphs, \textsc{Max}-$(\Prop_1, \Prop_2)$-\textsc{Partition With $D_2$} is equivalent to \textsc{Dominating Clique}.
    It is known that various problems can be expressed in this form, see~\cite{TamuraIZ20,TelleP97} in details.

    Bui-Xuan et al.\ demonstrated that, if each of $\Prop_1$ and $\Prop_2$ contains all graphs and $D_k[i,j]$ for each pair of $i,j \in \segsingle{k}$ is a finite or co-finite subset of $\mathbb{N}$, then \textsf{R}-$(\Prop_1, \Prop_2,\ldots, \Prop_k)$-\textsc{Partition With $D_k$} is solvable in polynomial time for bounded mim-width graphs by using the $d$-neighbor equivalence (assuming that a branch decomposition with constant mim-width is given)~\cite{BuixuanTV13}.
    The algorithm in \Cref{sec:polar} suggests that \textsf{R}-$(\Prop_1, \Prop_2,\ldots, \Prop_k)$-\textsc{Partition} is tractable if an equivalence relation $\equiv^i$ for each property $\Prop_i$ is defined appropriately and a family of representatives defined by $\equiv^i$ has polynomial size.
    Moreover, in \Cref{sec:polar,sec:con}, we have described how to deal with a complementary property and a connected variant on a branch decomposition with mim-width at most~$1$.
    In conclusion, combining with the work in~\cite{BuixuanTV13}, we obtain the following theorem.

    \begin{theorem} \label{the:mimone_solvability}
        Let $k$ be a fixed integer, $\Prop_1, \Prop_2,\ldots, \Prop_k$ be fixed graph properties, and $D_k$ be a degree constraint matrix such that $D_k[i,j]$ for each $i,j \in \segsingle{k}$ is a finite or co-finite subset of $\mathbb{N}$.
        Suppose that for each $i \in \segsingle{k}$, $\Prop_i$ or $\compProp_i$ is the collection of all graphs or all cluster graphs.
        Then, given a graph and its branch decomposition with mim-width at most~$1$, there is a polynomial-time algorithm that solves \textsf{R}-$(\Prop_1, \Prop_2,\ldots, \Prop_k)$-\textsc{Partition With $D_k$}, as well as the variant that some of $V_1, V_2, \ldots V_k$ are required to be connected.
    \end{theorem}

    Finally, we show \Cref{the:dichotomy} as an application of \Cref{the:mimone_solvability}.

    \begin{theorem} \label{the:dichotomy}
        All the following problems, as well as their connected variants and their dominating variants, are NP-hard for graphs with mim-width at most $2$:
  	(i) \textsc{Clique};
		(ii) \textsc{Induced Cluster Subgraph};
		(iii) \textsc{Induced Polar Subgraph};
		(iv) \textsc{Induced $\ovl{P_3}$-free Subgraph};
		(v) \textsc{Induced Split Subgraph}; and
        (vi) \textsc{Induced $\ovl{K_3}$-free Subgraph}.
        On the other hand, given a graph and its branch decomposition of mim-width at most $1$, all the above problems, as well as their connected variants and their dominating variants, are solvable in polynomial time.
    \end{theorem}

    \begin{proof}
        The NP-hardness of the above problems on graphs mim-width at most~$2$ was shown by \Cref{cor:NP-hard_lmimtwo}.

        It is not hard to see that \textsc{Clique}, \textsc{Induced Cluster Subgraph}, \textsc{Induced Polar Subgraph}, and \textsc{Induced $\ovl{P_3}$-free Subgraph} can be described in the forms of \textsf{R}-$(\Prop_1, \Prop_2,\ldots, \Prop_k)$-\textsc{Partition With $D_k$} such that the conditions in \Cref{the:mimone_solvability} are satisfied.
        In addition, \textsc{Induced Split Subgraph} is equivalent to \textsc{Min}-$(\Prop_1, \Prop_2, \Prop_3)$-\textsc{Partition With $D_3$} such that $\Prop_1$ and $\Prop_2$ are the collections of all graphs, $\Prop_3$ is the collection of all connected cluster graphs (that is, cliques), and $D_3$ is defined as follows:
        \[
            D_3 =
            \begin{pmatrix}
                \mathbb{N} & \mathbb{N} & \mathbb{N} \\
                \mathbb{N} & \{ 0 \} & \mathbb{N} \\
                \mathbb{N} & \mathbb{N} & \mathbb{N} \\
            \end{pmatrix}.
        \]
        The entry $D_3[2,2]$ requires $V_2$ to be an independent set of a given graph $G$.
        By \Cref{the:mimone_solvability}, the problems (i)--(v) are solvable in polynomial time.

        Our remaining task is to show the polynomial-time solvability of \textsc{Induced $\ovl{K_3}$-free Subgraph}.
        To this end, we provide the simple observation that \textsc{Induced $K_3$-free Subgraph} is solvable in polynomial time if a branch decomposition with mim-width at most~$1$ is given.
        Let $G$ be a graph with mim-width at most~$1$ and $S$ be a vertex subset of $G$.
        We claim that $\indG{S}$ is $K_3$-free if and only if $\indG{S}$ is bipartite.
        Since $G$ is perfect from \Cref{prop:mimone_perfect}, $G$ has no induced odd cycle of length at least~$5$.
        Thus, if $\indG{S}$ is $K_3$-free, then $\indG{S}$ has no induced odd cycle, that is, $\indG{S}$ is bipartite.
        Obviously, the converse is also true.
        Therefore, \textsc{Induced $K_3$-free Subgraph} is equivalent to \textsc{Induced Bipartite Subgraph}, which can be expressed in \textsc{Min}-$(\Prop_1, \Prop_2, \Prop_3)$-\textsc{Partition With $D_3^\prime$} such that $\Prop_1$, $\Prop_2$, and $\Prop_3$ are the collections of all graphs and $D_3^\prime$ is defined as follows:
       \[
            D_3^\prime =
            \begin{pmatrix}
                \mathbb{N} & \mathbb{N} & \mathbb{N} \\
                \mathbb{N} & \{ 0 \} & \mathbb{N} \\
                \mathbb{N} & \mathbb{N} & \{ 0 \} \\
            \end{pmatrix}.
        \]
        In conclusion, reducing \textsc{Induced $\ovl{K_3}$-free Subgraph} to \textsc{Induced $K_3$-free Subgraph} by taking the complement of a given graph, from \Cref{prop:mimone_comp} and \Cref{the:mimone_solvability}, \textsc{Induced $\ovl{K_3}$-free Subgraph} is solvable in polynomial time if a branch decomposition with mim-width at most~$1$ is given.
        The connected variant and dominating variant of the problems (i)--(vi) are also shown to be solvable in polynomial time by similar arguments.
    \end{proof}

    \section{Concluding remarks}

    We discuss future work here.
    Our proof of \Cref{the:NP-hard_main} relies on the assumption that a fixed graph property $\Prop$ admits all cliques, and hence \Cref{the:NP-hard_main} is not applicable to \textsc{Induced $\Prop$ Subgraph} such that $\Prop$ excludes some clique.
    Such problems include \textsc{Independent Set}, \textsc{Induced Matching}, \textsc{Longest Induced Path}, and \textsc{Feedback Vertex Set}.
    In fact, there exist XP algorithms of the problems listed above when parameterized by mim-width~\cite{BelmonteV13,BuixuanTV13,JaffkeKT20a,JaffkeKT20b}.
    This motivates us to seek a graph property $\Prop$ such that \textsc{Induced $\Prop$ Subgraph} is NP-hard for bounded mim-width graphs although $\Prop$ excludes some clique.
    As the first step, it would be interesting to consider \textsc{Induced $K_3$-free Subgraph}.

    In~\cite{BergougnouxDJ23}, Bergougnoux et al.~showed that \textsc{Clique} is expressible in $\textsf{A}\&\textsf{C DN}+\forall$, which is \textsf{A}\&\textsf{C DN} logic that allows to use a single universal quantifier $\forall$, and hence their meta-theorem cannot be extended to $\textsf{A}\&\textsf{C DN}+\forall$.
    Our results in this paper suggest that the barrier could be broken down for graphs with mim-width at most~$1$.
    The next goal is to obtain a more general logic than \textsf{A}\&\textsf{C DN} such that all problems expressible in the logic are solvable in polynomial time for graphs with mim-width at most~$1$.

    Finally, we end this paper by leaving the biggest open problem concerning mim-width: Given a graph $G$, is there a polynomial-time algorithm that computes a branch decomposition with mim-width~$1$, or concludes that $G$ has mim-width more than~$1$?



	\bibliography{reference}


\end{document}